\documentclass[11pt]{article}  
\usepackage{amsmath,amssymb, amsfonts,amsthm}
\usepackage[T1]{fontenc}
\usepackage[utf8]{inputenc}
\usepackage{tikz}
\usepackage{epsfig}
\usepackage{todonotes}
\usetikzlibrary{shapes.geometric}
\usepackage[normalem]{ulem}

\def\DIRCONV{_\textrm{DIR}\textrm{-CONV}}
\def\DIRCONVd{^d_\textrm{DIR}\textrm{-CONV}}
\newcommand{\DIRCONVk}[1]{{_{\textrm{DIR}(#1)}\textrm{-CONV}}}
\newcommand{\DIRCONVkd}[1]{{^d_{\textrm{DIR}(#1)}\textrm{-CONV}}}

\def\rr{\mathbb{R}}
\def\L{\mathcal{L}}
\def\P{\mathcal{P}}

\def\phom{{P_{hom}}}

\newtheorem{theorem}{Theorem}
\newtheorem{lemma}{Lemma}
\newtheorem{corollary}{Corollary}
\newtheorem{proposition}{Proposition}
\theoremstyle{definition}

\newtheorem{problem}{Problem}
\newtheorem{conjecture}{Conjecture}

\long\def\smazat#1{}

\begin{document}
\title{Homothetic Polygons and Beyond: Intersection Graphs, Recognition and Maximum Clique}
%

\author{
Valentin E. Brimkov\footnote{Mathematics Department, SUNY Buffalo State, Buffalo, NY 14222, USA.
}
 \and
Konstanty Junosza-Szaniawski\footnote{Warsaw University of Technology, Faculty of Mathematics and Information Science, Koszykowa 75, 00-662 Warszawa, Poland. 
}
 \and
Sean Kafer\footnote{Department of Combinatorics and Optimization, University of Waterloo, Waterloo, Ontario N2L 3G1, Canada. 
}
 \and
 Jan Kratochv\'il \footnote{Department of Applied Mathematics, Faculty of Mathematics and Physics, Charles University, Malostransk\'e n\'am. 25, 118 00 Praha 1, Czech Republic. 
 }
\and
Martin Pergel\footnote{Department of Software and Computer Science Education, Faculty of Mathematics and Physics, Charles University, Malostransk\'e n\'am. 25, 118 00 Praha 1, Czech Republic. 
}
\and
Pawe{\l} Rz{\k a}{\. z}ewski\footnotemark[2]
\and
Matthew Szczepankiewicz\footnote{Mathematics Department, University at Buffalo, Buffalo, NY 14260-2900, USA.}
\and
Joshua Terhaar\footnotemark[6]
}

\date{ }


\maketitle

\begin{abstract}
We study the {\sc Clique} problem in classes of intersection graphs of convex sets in the plane. The problem is known to be NP-complete in convex-set intersection graphs and straight-line-segment intersection graphs, but solvable in polynomial time in intersection graphs of homothetic triangles. We extend the latter result by showing that for every convex polygon $P$ with sides parallel to $k$ directions, every $n$-vertex graph which is an intersection graph of homothetic copies of $P$ contains at most $n^{k}$ inclusion-wise maximal cliques. We actually prove this result for a more general class of graphs, the so called $k\DIRCONV$, which are intersection graphs of convex polygons whose sides are parallel to some fixed $k$ directions. Moreover, we provide some lower bounds on the maximum number of maximal cliques, discuss the complexity of recognizing these classes of graphs and present relationships with other classes of convex-set intersection graphs. Finally, we generalize the upper bound on the number of maximal cliques to intersection graphs of higher-dimensional convex polytopes in Euclidean space.
\end{abstract}
\section{Introduction}
Geometric representations of graphs, and intersection graphs in particular, are widely studied both for their practical applications and motivations, and for their interesting theoretical and structural properties. It is often the case that optimization problems, that are NP-hard for general graphs, can be solved, or at least approximated, in polynomial time on geometric intersection graphs. Classical examples are the {\sc Independent set}, {\sc Clique}, or {\sc Coloring} problems for interval graphs, one of the oldest intersection-defined classes of graphs \cite{interval}. The former two problems remain polynomially solvable in circle and polygon-circle graphs, while the last one already becomes NP-complete. For definitions and more results about these issues, as well as some possible applications, the interested reader is referred to the works of Golumbic \cite{Golumbic}, McKee and McMorris \cite{McKMcM}, and Spinrad \cite{Spinrad}.

In this paper we 
investigate subclasses of the class of intersection graphs of convex sets in the plane, denoted by CONV, and the computational complexity of the problem of finding a maximum clique in such graphs. This has been motivated  by three arguments. First, the {\sc Clique} problem was shown to be polynomial time solvable for intersection graphs of homothetic triangles in the plane by Kaufmann {\em et al.} \cite{SODA}. (These graphs have been shown to be equivalent to the so called max-tolerance graphs, and as such found direct application in DNA sequencing.) Secondly, the {\sc Clique} problem is known to be NP-complete in CONV graphs \cite{KratKub}, and so it is interesting to inspect the border between easy and hard instances more closely. Thirdly, straight-line segments are the simplest convex sets, and it is thus natural to ask how difficult {\sc Clique} is in intersection graphs of segments in the plane (this class is denoted by SEG). Kratochv\'{\i}l and Ne\v{s}et\v{r}il posed this problem in \cite{KN} after they observed that if the number of different directions of the segments is bounded by a constant, say $k$, a maximum clique can be found in time $O(n^{k+1})$ (this class of graphs is denoted by $k$-DIR, see \cite{KratMat} for more details). This question was answered very recently by Cabello {\em et al.} \cite{Cab} who showed that {\sc Clique} is NP-complete in SEG graphs. Maximal and maximum cliques in intersection graphs of convex sets have also been considered by Amb\"uhl and Wagner \cite{AW} (ellipses and triangles), Brimkov {\em et al.} \cite{BKST} (trapezoids), and Imai and Asano \cite{IA} (rectangles).

In \cite{KratPer}, Kratochv\'{\i}l and Pergel initiated a study of $P_{hom}$ graphs, defined as intersection graphs of convex polygons homothetic to a given polygon $P$. They announced that for every convex polygon $P$, recognition of $P_{hom}$ graphs is NP-hard, and asked in Problem~3.1 if $P_{hom}$ graphs can have a superpolynomial number of maximal cliques. Our main result shows that for every convex $k$-gon $P$, every $P_{hom}$ graph with $n$ vertices contains at most $n^{k}$ maximal cliques, and hence {\sc Clique} is solvable in polynomial-time on $P_{hom}$ graphs for every fixed polygon $P$. For the sake of completeness, we also present the proof of NP-hardness of $P_{hom}$ recognition.

E.J. van Leeuwen and J. van Leeuwen \cite{vLvL} considered a more general class of graphs based on affine transformations of one (or more) master objects.
These were called $\cal P$-intersection graphs, where ${\cal P} = (S,T)$ is a signature consisting of a set $S$ of master objects  and a set $T$ of transformations. They proved that if all objects in the signature are described by rational numbers, such graphs have representations of polynomial size and the recognition problem is in NP. As a corollary, recognition of $P_{hom}$ graphs is in NP (and hence NP-complete) for every rational polygon $P$. In \cite{vLvLM}, van Leeuwens and T. M\"uller proved tight bounds on the maximum sizes of representations (in terms of coordinate sizes) of $P_{translate}$ (i.e. intresection graphs of translated copies of some convex polygon $P$) and $P_{hom}$ graphs.

In proving the main result of our paper, the polynomial bound on the number of maximal cliques, we go beyond the homothetic polygon intersection graphs. We observe that in any representation by polygons homothetic to a master one, the sides of the polygons are parallel to a bounded number of directions in the plane. So, if we relax the requirement on the homothetic relation of the polygons in the representation, we simply consider a set of $k$ directions and look after graphs that have intersection representations by convex polygons, whose every side is parallel to one of those $k$ directions (see Figure \ref{examples}). We call this class $k\DIRCONV$ graphs.

We investigate the class of $k\DIRCONV$ graphs and discuss the complexity of its recognition and relationships to other relevant graph classes (SEG, $k$-DIR, and $P_{hom}$). 
We prove that every such graph has at most $n^{k}$ maximal cliques, where $k$ is the number of directions parallel to at least one side (since we may have two parallel sides). We find this fact worth emphasizing, as it also covers van Leeuwens' $\cal P$-intersection graphs for transformations without rotations. 
The immediate complexity impact of this result is that, for every convex polygon $P$, the {\sc Clique} problem can be solved in polynomial time in $k\DIRCONV$ graphs, even when a representation of the input graph is not given. (It is well-known that all maximal cliques of an input graph can be enumerated with polynomial delay, see Tsukiyama {\em et al.} \cite{maxcliques}.) The exponent of the polynomial of course depends on $k$.

We also pay closer attention to maximal cliques in $P_{hom}$ graphs for specific polygons $P$. If $P$ is a regular $k$-gon, then we can construct a $P_{hom}$ graph with $\Omega(n^{\lfloor k/2 \rfloor (1-\epsilon)})$ maximal cliques (where $\epsilon$ is an arbitrarily small positive constant).

Moreover, for every fixed polygon $P$ but parallelograms we present a construction of a $P_{hom}$  graph with $\Omega(n^3)$ maximal cliques (by a modification of a construction for triangles from \cite{SODA}). It is worth noting that also for the max-coordinate results of \cite{vLvLM}, parallelograms play an exceptional role.

Finally, we generalize the upper bound on the number of maximal cliques to intersection graphs of convex polytopes of higher dimensions. Intersection graphs of higher-dimensional polytopes (namely, high-dimensional boxes) have been considered earlier (see for example \cite{asplund,kostochka,Roberts}).

\begin{figure}[ht]
\begin{center}
\begin{tikzpicture}[xscale=0.5, yscale=0.5]
\newcommand*\poly{--++(2,0)--++(1,1)--++(-3,2)--++(-1,-2) -- cycle}
\draw (0,0) \poly;
\draw[scale=1.2] (2,-1) \poly;
\draw[scale=0.5] (5,4) \poly;
\end{tikzpicture}
\hskip 20 pt
%
\begin{tikzpicture}[yscale=0.5, xscale=0.5]

\draw (0,0)--++(2,0)--++(1,1)--++(-3,2)--++(-1,-2) -- cycle;
\draw[dotted]   (-5,0) -- (6,0)
                (0,-2) -- (5,3)
                (6,-1) -- (-3,5)
                (1,5) -- (-2,-1)
                (-3,3) -- (3,-3);

\draw[scale=0.6] (-3.5,-4) --++ (8,8)--++(-3,2) -- cycle;
\draw[dotted]   (-5.5 * 0.6, -6 * 0.6) -- (6.5 * 0.6,6 * 0.6)
                (7.5 * 0.6, 2 * 0.6) -- (-1.5 * 0.6,8 * 0.6)
                (3 * 0.6,9 * 0.6) -- (-5 * 0.6,-7 * 0.6);

\draw (-1.5,-2) --++ (1,1)--++(-3,2)--++(-1,-1) -- cycle;
\draw[dotted]   (-1.5 - 2,-2 - 2) -- (-0.5 + 5,-1 + 5)
                (-0.5 + 3,-1 -2) -- (-3.5 - 3,1 + 2)
                (-3.5 +5 ,1 +5 ) -- (-4.5 -2 ,0 -2 )
                (-4.5 -3 ,2) -- (-1.5 +3 ,-2 -2 ) ;

\end{tikzpicture}
\end{center}
\caption{Homothetic pentagons (left) and polygons with 5 directions of sides (right).}
\label{examples}
\end{figure}
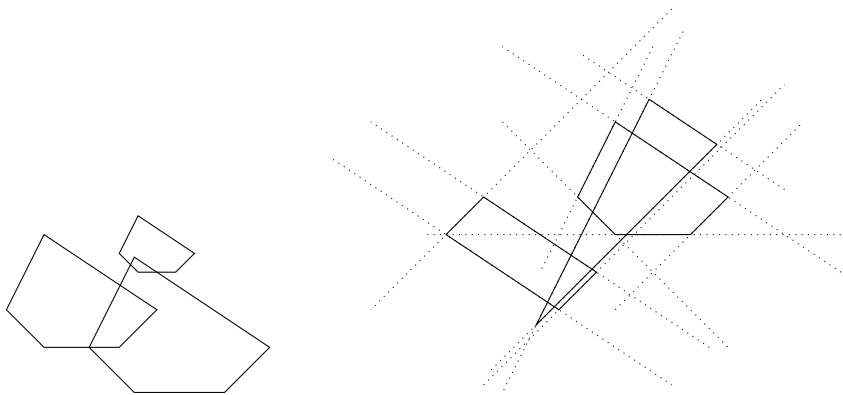

\section{Definitions and basic properties}

In this paper we deal with intersection graphs of subsets of the Euclidean plane $\rr^2$. The following concepts are standard and we only briefly overview them to make the paper self-contained. For a collection $R$ of sets  the {\em intersection graph} of $R$ is denoted by $IG(R)$; its vertices are in 1-1 correspondence with the sets and two vertices are adjacent if and only if the corresponding sets are non-disjoint. In such a case the collection $R$ is called an {\em (intersection) representation} of $G$, and the set corresponding to a vertex $v \in V(G)$ is called {\em the representative} of $v$ and denoted by $R_v$.

The intersection graphs of straight-line segments are called the SEG graphs, of convex sets the CONV graphs, and $k$-DIR is used for SEG graphs having a representation with all the segments being parallel to at most $k$ directions (thus 1-DIR are exactly the interval graphs). 
For a fixed set $P$   (in most cases a convex polygon), the class of intersection graphs of sets homothetic to $P$ is denoted by $P_{hom}$ (two sets are homothetic if one of them can be obtained from the other by scaling and/or translating). If $P$ is a disk, we get disk-intersection graphs, a well studied class of graphs. {\em Pseudodisk intersection graphs} are intersection graphs of collections of closed planar regions (bounded by simple Jordan curves) that are pairwise in a {\em pseudodisk} relationship, i.e., both differences $A\setminus B$ and $B\setminus A$ are arc-connected.

{
We say that a system of arc-connected closed subsets of the plane is in the {\em general position} if:
\begin{itemize}
\item there are no two sets whose boundaries share infinitely many points;
\item there are no three sets whose boundaries share a common point;
\item 
let $x$ be the intersection point of the boundaries of sets $A$ and $B$.
There is a closed neighborhood of $x$, such that when moving on its boundary in clock-wise direction, the boundaries of $A$ and $B$ are met alternately.  

\end{itemize}
Intuitively, the general position forbids the situations depicted in Figure \ref{general}. Notice that if we allow translations and scaling, then we can rearrange any system of arc-connected sets so that they are in general position and the intersections do not change (and thus the intersection graph remains the same).

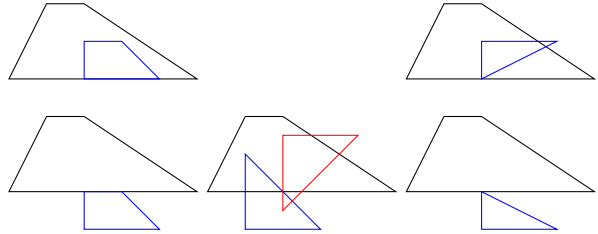
\begin{figure}[ht]
\begin{center}
\begin{tikzpicture}[xscale=0.5, yscale=0.5]
\draw[color=black] (0,0) --++ (5,0) --++ (-3,2) --++ (-1,0) -- cycle;
\draw[color=blue] (2,0) --++ (2,0) --++ (-1,1) --++ (-1,0) -- cycle;
\draw[color=black] (0,-3) --++ (5,0) --++ (-3,2) --++ (-1,0) -- cycle;
\draw[color=blue] (2,-4) --++ (2,0) --++ (-1,1) --++ (-1,0) -- cycle;
\end{tikzpicture}
\begin{tikzpicture}[xscale=0.5, yscale=0.5]
\draw[color=black] (0,0) --++ (5,0) --++ (-3,2) --++ (-1,0) -- cycle;
\draw[color=blue] (1,-1) --++ (2,0) --++ (-2,2) -- cycle;
\draw[color=red] (2,-0.5) --++ (2,2) --++ (-2,0) -- cycle;
\end{tikzpicture}
\begin{tikzpicture}[xscale=0.5, yscale=0.5]
\draw[color=black] (0,0) --++ (5,0) --++ (-3,2) --++ (-1,0) -- cycle;
\draw[color=blue] (2,0) --++ (2,1) --++ (-2,0) -- cycle;
\draw[color=black] (0,-3) --++ (5,0) --++ (-3,2) --++ (-1,0) -- cycle;
\draw[color=blue] (2,-4) --++ (2,0) --++ (-2,1) -- cycle;
\end{tikzpicture}
\end{center}
\caption{Situation forbidden in the general position.}
\label{general}
\end{figure}

Another well-known property is that two sets in general position are in a pseudodisk relation if and only if their boundaries do not intersect or intersect twice.}

Throughout the paper, a {\em polygon} means a closed convex polygon in the plane. 
Let $\L$ be the set of all distinct lines in $\rr^2$ that contain the point $(0,0)$.  For a $k$-tuple of lines $L = \{ \ell_1,..,\ell_k \} \in \binom{\L}{k}$, we denote by $\P(L)$ the family of all polygons $P$ such that every side of $P$ is parallel to some $\ell \in L$. Moreover, by $\P(k)$ we denote $\bigcup_{\binom{\L}{k}} \P(L)$.

Now we introduce the main characters of the paper.
By $k\DIRCONVk{L}$ we denote the class of intersection graphs of polygons of $\P(L)$. Finally, we define $k\DIRCONV = \bigcup _{L \in \binom{\L}{k}} k\DIRCONVk{L}$. Figure~\ref{examples} shows examples of representations of the same graph $P_3$ by intersections of homothetic pentagons and as a $5\DIRCONV$ graph. 

Note that polygons in $\P(L)$ do not have to be in  a pseudodisk relation while homothetic copies of the same polygon always are.
Moreover, observe that $2\DIRCONV$ are the intersection graphs of isothetic rectangles, which are exactly the graphs of boxicity at most 2 (see \cite{Roberts} for more details on boxicity of graphs).

The following property of convex polygons is well-known.
\begin{lemma} [Folklore]\label{lem:separate2}
Any two disjoint convex polygons in $\P(L)$ can be separated by a line parallel to a line from $L$.
\end{lemma}

\section{Relations between graph classes}

In this section we investigate the relations between the graph classes considered in this paper, i.e., $k\DIRCONV$, $P_{hom}$, SEG, and $k$-DIR.
We first observe that for every $k \geq 2$, each $k$-DIR graph is also in $k\DIRCONV$.

\begin{theorem} \label{kdir-kdirconv}
For every $k \geq 2$, it holds that $k$-DIR $\subseteq k\DIRCONV$.
\end{theorem}
\begin{proof}
Let $G$ be a $k$-DIR graph and let $R=\{S_i \colon i\in\{1,..,n\}\}$ be a segment representation of $G$. Let $L=\{\ell_1,..,\ell_k\}$ be a set of lines such that every segment in $R$ is parallel to some line in $L$.
We will define a family $R'$ of parallelograms from $\P(L)$, such that the intersection graph of $R'$ is isomorphic to $G$. This will show that $G \in k\DIRCONVk{L}$ and therefore $G \in k\DIRCONV$.
The idea of constructing $R'$ is to extend every segment from $R$ to a very narrow parallelogram in such a way that no new intersection appears (see Figure \ref{transform-kdir}).

\begin{figure}[ht]
\begin{center}
\begin{tikzpicture}[scale=0.7]
\draw (0,0) -- (0,4);
\draw (-1,3) -- (3,3);
\draw (-1,0) -- (3,4);
\draw (0.5,2) -- (3,2);
\draw (3.5,0.5) -- (3.5,3.5);  
\node (p) at (1,-0.5) {$R$}; 
\end{tikzpicture}
\hskip 40pt
\begin{tikzpicture}[scale=0.7]
\draw (0,0) -- (0,4) -- (0.2,4) -- (0.2,0) -- (0,0);
\draw (-1,3) -- (3,3) -- (3.2,3.2) -- (-0.8,3.2) -- (-1,3);
\draw (-1,0) -- (3,4) -- (3,4.2) -- (-1,0.2) -- (-1,0);
\draw (0.5,2) -- (3,2) -- (3.2,2.2) -- (0.7,2.2) -- (0.5,2);
\draw (3.5,0.5) -- (3.5,3.5) -- (3.7, 3.5) -- (3.7, 0.5) -- (3.5,0.5);
\node (p) at (1,-0.5) {$R'$}; 
\end{tikzpicture}

\end{center}
\caption{Transformation from a segment representation (left) to a polygon representation (right).}
\label{transform-kdir}
\end{figure}
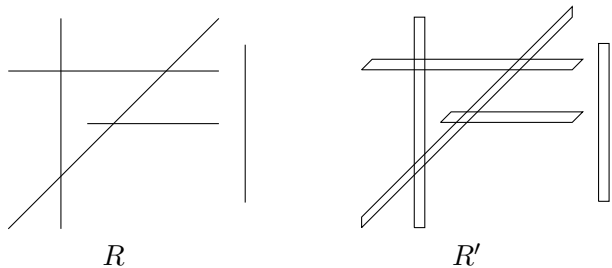

Let $d=\min\{\mathrm{dist}(S_i,S_j): S_i\cap S_j=\emptyset, i,j\in \{1,..,n\}\}$ (where $\mathrm{dist}(S_i,S_j)$ denotes the length of the shortest segment with one end in $S_i$ and the other in $S_j$) and
let $x_i$ and $y_i$ be the endpoints of the segment $S_i$ for $i\in \{1,..,n\}$. Consider $S_i$ for $i\in \{1,..,n\}$. The segment $S_i$ is parallel to a line from $L$, say to $\ell_j$. Let $b_i$ be a unit length vector parallel to $\ell_{j+1 \mod k}$. We set $P_i$ to be the parallelogram with corners $x_i+\frac{d}{3}b_i, x_i-\frac{d}{3}b_i,y_i-\frac{d}{3}b_i, y_i+\frac{d}{3}b_i$, and we set $R'=\{P_i:i\in \{1,..,n\}\}$.

Obviously, if $S_i\cap S_j\neq \emptyset$, then $P_i\cap P_j\neq \emptyset$, since $S_i$ is contained in $P_i$ for every $i\in \{1,..,n\}$. On the other hand notice that every point of $P_i$ is at distance at most $\frac{d}{3}$ from $S_i$. Assume $P_i\cap P_j\neq \emptyset$ and consider a point $z \in P_i\cap P_j$. Such a $z$ is at distance at most $\frac{d}{3}$ from $S_i$ and $S_j$. By the triangle inequality, $S_i$ and $S_j$ are at distance at most $\frac{2d}{3}$ and by the definition of $d$ they intersect each other.
\end{proof}

Now let us turn our attention to $P_{hom}$ graphs.
Clearly $P_{hom} \in k\DIRCONV$ for any convex polygon $P$ with sides parallel to at most $k$ directions. In particular, $P_{hom} \in k\DIRCONV$ for any $k$-gon $P$.

Next we prove that $P_{hom}$ graphs are pseudodisk intersection graphs for any $P$. Although the proof is claimed to be known, e.g., in \cite{APS}, we enclose an alternative proof to make the paper self-contained.

{
\begin{lemma} \label{lemma:homarepdisk}
Homothetic convex bodies in the plane in general position form an arrangement
of pseudodisks. Thus, for every convex polygon $P$, the class $P_{hom}$ is
a subclass of the class of pseudodisk intersection graphs.
\end{lemma}
\begin{proof}
Let $A$ and $B$ be two homothetic convex bodies in general position.
First we consider the case where one of them, say $A$, is smaller than the other. 
We use the Banach fixed point theorem
for a linear mapping which maps the bigger polygon to the smaller one
(such a mapping can be obtained as a composition of shift and scaling). 
The Banach theorem implies that such a mapping has a  fixed point $f$ (note that we do not need an efficient algorithm to find one). If $f$ is inside the two bodies, then $A$ lies entirely inside $B$ and therefore $A$ and $B$ are in a pseudodisk relation.
If $f$ lies on the common boundary of $A$ and $B$, then the boundaries of $A$ and $B$ share either one or infinitely many points, which contradicts 
the condition that $A$ and $B$ are in general position.

Now consider the case where $f$ lies outside the bodies.
Consider a half-line $\ell$ originating at $f$ and intersecting both $A$ and $B$.
As $A$ and $B$ are homothetic and $f$ is the fixed point of the mapping, $\ell$ encounters the boundaries of $A$ and $B$ in the following order:
\begin{enumerate}
\item The start-intersection-point of $A$;
\item The start-intersection-point of $B$ or the end-intersection-point of $A$;
\item Once again, either the start-intersection-point of $B$ or the end-intersection-point of $A$; 
\item The end-intersection-point of $B$. 
\end{enumerate}
\begin{figure}
\hfill \includegraphics{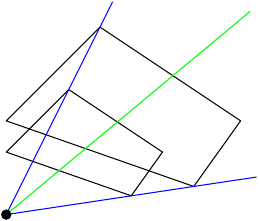}
\hfill\hfill \includegraphics{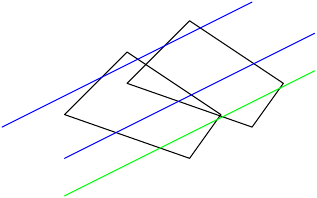}\hfill~
\caption{
Illustration to the procedure of  finding a fixed point (left) and to the behaviour of that fixed point ``in infinity'' for bodies of the same size (right).}
\end{figure}

Suppose that $A$ and $B$ are not in pseudo-disk relation, i.e., their boundaries share at least three points $x_1,x_2$ and $x_3$ (the case where the boundaries do not intersect was already considered, and the case where there is just one intersection point is forbidden for sets in general position).

Let $\ell_1,\ell_2$ and $\ell_3$ be the half-lines originating at point $f$ and containing $x_1, x_2$ and $x_3$, respectively. Moreover, assume that 
$\ell_1$ is the first and $\ell_3$ the last of these three half-lines with respect to clockwise ordering about $f$.

Consider the order in which $\ell_i$ (for $i=1,2,3$) encounters the boundaries of $A$ and $B$. Since $x_i$ belongs to the boundaries of both $A$ and $B$, and the bodies are in the general position (which implies they do not share a segment of a boundary), this means that $X_i$ is the point encountered in parts 2 and 3. Thus $\ell_i$ first encounters the start-intersection-point of $A$, then the point $x_i$ (which is simultaneously the end-intersection point of $A$ and the start-intersection-point of $B$) and then finally the end-intersection-point of $B$.

Assume first that $x_1,x_2$, and $x_3$ are co-linear. Since $A$ and $B$ are convex, the segment $x_1x_3$ belongs to both $A$ and $B$ and forms their common border. But this means that the boundaries of $A$ and $B$ share infinitely many points, which contradicts the general position condition.

Otherwise, if $x_1, x_2$, and $x_3$ do not lie on the same straight line, the segment $x_1x_3$ is contained in both $A$ and $B$ (by convexity). Let $x$ be the intersection point of $\ell_2$ and $x_1x_3$. Since $x_2$ is the end-intersection-point of $A$ and the start-intersection-point of $B$, $x$ does not belong to at least one of the sets $A$ and $B$ -- a contradiction (see Figure \ref{fig-pseudo}).

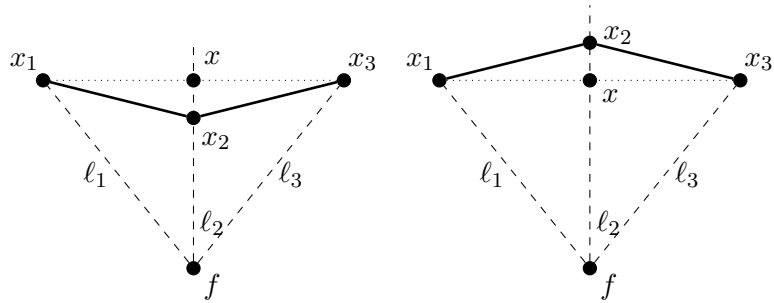
\begin{figure}[ht]
\begin{center}
\begin{tikzpicture}[xscale=0.5, yscale=0.5]
\node[draw,circle,fill=black, inner sep=0pt,minimum size=5pt] (f) at (0,0) {};
\node at (0.5,-0.5) {$f$};
\node[draw,circle,fill=black, inner sep=0pt,minimum size=5pt] (x1) at (-4,5) {};
\node at (-4.5,5.5) {$x_1$};
\node[draw,circle,fill=black, inner sep=0pt,minimum size=5pt] (x2) at (0,4) {};
\node at (0.6,3.4) {$x_2$};
\node[draw,circle,fill=black, inner sep=0pt,minimum size=5pt] (x3) at (4,5) {};
\node at (4.5,5.5) {$x_3$};
\node[draw,circle,fill=black, inner sep=0pt,minimum size=5pt] (x) at (0,5) {};
\node at (0.5,5.6) {$x$};
\draw[dashed] (f) --node[below, left]{$\ell_1$} (x1);
\draw[dashed] (f) -- (0,6);
\node at (0.5, 1.2) {$\ell_2$};
\draw[dashed] (f) --node[right]{$\ell_3$} (x3);
\draw[line width=1] (x1) -- (x2);
\draw[line width=1] (x2) -- (x3);
\draw[dotted] (x1) -- (x3);
\draw[dashed] (x2) -- (x3);
\end{tikzpicture}
\begin{tikzpicture}[xscale=0.5, yscale=0.5]
\node[draw,circle,fill=black, inner sep=0pt,minimum size=5pt] (f) at (0,0) {};
\node at (0.5,-0.5) {$f$};
\node[draw,circle,fill=black, inner sep=0pt,minimum size=5pt] (x1) at (-4,5) {};
\node at (-4.5,5.5) {$x_1$};
\node[draw,circle,fill=black, inner sep=0pt,minimum size=5pt] (x2) at (0,6) {};
\node at (0.75,6.2) {$x_2$};
\node[draw,circle,fill=black, inner sep=0pt,minimum size=5pt] (x3) at (4,5) {};
\node at (4.5,5.5) {$x_3$};
\node[draw,circle,fill=black, inner sep=0pt,minimum size=5pt] (x) at (0,5) {};
\node at (0.55,4.6) {$x$};
\draw[dashed] (f) --node[below, left]{$\ell_1$} (x1);
\draw[dashed] (f) -- (0,7);
\node at (0.5, 1.2) {$\ell_2$};
\draw[dashed] (f) --node[right]{$\ell_3$} (x3);
\draw[line width=1] (x1) -- (x2);
\draw[line width=1] (x2) -- (x3);
\draw[dotted] (x1) -- (x3);
\draw[dashed] (x2) -- (x3);
\end{tikzpicture}
\end{center}
\caption{The case where $x_1,x_2$, and $x_3$ are not co-linear. The body $A$ lies below $x_2$ and $B$ lies above $x_2$ in the direction of $\ell_2$. Thus $x \notin A$ (left) or $x \notin B$ (right).}
\label{fig-pseudo}
\end{figure}

The case of two convex bodies of the same size can be handled analogously. 
The only somewhat more essential difference is that instead of half-lines originating at $f$, one can use  parallel lines in the appropriate direction.
All arguments remain essentially the same, up to the fact that the ``fixed point'' would lie in the infinity.
This completes the  proof.
\end{proof}
}

\begin{lemma} \label{kdir-phom}
Graph classes $k$-DIR and $P_{hom}$ are essentially distinct, i.e., for any polygon $P$ and 
$k \geq 2$, 
\begin{enumerate}
 \item $k$-DIR $\not\subseteq P_{hom}$; 
 \item $P_{hom} \not\subseteq k$-DIR.
\end{enumerate}
\end{lemma}

\begin{proof}
\begin{enumerate}
\item Let us consider the graph $K_{3,3}$. Figure \ref{k33} shows that it is in 2-DIR (and thus in $k$-DIR for any $k \geq 2$). On the other hand, Kratochv\'{i}l \cite{Krat-pseudo} proved that triangle-free pseudodisk intersection graphs are planar. Since by Lemma \ref{lemma:homarepdisk} all $P_{hom}$ graphs are pseudodisk  intersection graphs, $K_{3,3}$ is not a $P_{hom}$ graph for any polygon $P$.
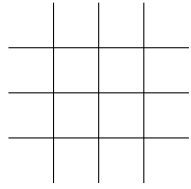
\begin{figure}[htbp]
\centering
\begin{tikzpicture}[scale=0.6]
\draw (0,1) -- (4,1);
\draw (0,2) -- (4,2);
\draw (0,3) -- (4,3);
\draw (1,0) -- (1,4);
\draw (2,0) -- (2,4);
\draw (3,0) -- (3,4);
\end{tikzpicture}
\caption{A 2-DIR representation of $K_{3,3}$. It can be clearly generalized for any $K_{n,n}$.}
\label{k33}
\end{figure}

\item Let $k$ be fixed and let $S_{k}$ be a graph consisting of a $(2k+1)$-clique, whose every vertex has a private neighbor.
Figure \ref{sun} shows $S_{k}$ and its geometric representation as a $P_{hom}$ graph for $P$ being a square and a triangle. Note that this construction can be easily generalized for any convex polygon $P$. Thus $S_k \in P_{hom}$ for any $P$.

Suppose now that $S_k \in k$-DIR. Since we have $2k+1$ vertices in the clique and only $k$ directions available, there exists a direction with at least three segments in the clique. Since those segments are pairwise parallel and intersect each other, they lie on the same line. But in such a case at most two of them may have private neighbors. Therefore $S_k \notin k$-DIR.
\end{enumerate}
\end{proof}

\begin{figure}[ht]
\centering
\begin{tikzpicture}[scale=0.6]
\tikzstyle{every node}=[draw,shape=circle,fill=white];

\foreach \i in {0,1,...,4}
{	
	\draw (72*\i:1cm)--(72*\i:2cm);

\foreach \j in {0,1,...,\i}
{	
	\draw (72*\i:1cm)--(72*\j:1cm);
}	
}
\foreach \i in {0,1,...,4}
{
    \draw[fill = black] (72*\i:1cm) circle (3 pt);
    \draw[fill = black] (72*\i:2cm) circle (3 pt);
}
\end{tikzpicture}
\hskip 50 pt
\begin{tikzpicture}[yscale=0.3, xscale=0.3]
\newcommand*\nn{5}

\foreach \x in {1,...,\nn}
{
    \draw (-1.0 + \x, 0.0 + \x/2) --++ (6,0) --++ (0,6) --++ (-6,0) -- cycle;
    \draw (-1.3 + \x, 5.7 + \x/2) --++ (0.6,0) --++ (0,0.6) --++ (-0.6,0) -- cycle;
}
\end{tikzpicture}
\hskip 50 pt
\begin{tikzpicture}[yscale=0.3, xscale=0.3]
\newcommand*\nn{5}

\foreach \x in {1,...,\nn}
{
    \draw (-1.0 + \x, 0.0 + \x/4) --++ (6,0) --++ (-6,8) -- cycle;
    \draw (-1.1 + \x, 7.8 + \x/4) --++ (0.6,0) --++ (-0.6,0.8) -- cycle;
}
\end{tikzpicture}
\caption{Graph $S_2$ and its representation by squares and by homothetic triangles.}
\label{sun}
\end{figure}
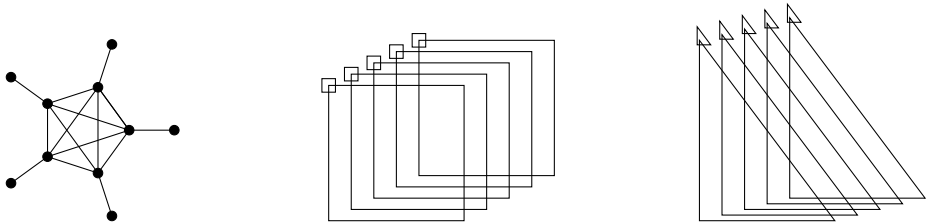

The construction from the second part of Lemma \ref{kdir-phom} exploited the fact that the number of directions of segments was fixed.
For $k\DIRCONV$ graphs we can improve the result by constructing a graph that is in $k\DIRCONV$ for any $k \geq 2$, but cannot be represented by any configuration of segments (so it is not in SEG).

\begin{theorem} \label{thm:counter}
For any $k \geq 2$ there exists a graph in $k\DIRCONV$ which is not a SEG graph.
\end{theorem}

\begin{proof}
Let us consider the graph $G$ in Figure \ref{counterexample2} which is inspired by construction of Kratochv\'{i}l and Matou\v{s}ek from \cite{KratMat}.
Suppose that $G$ is a SEG graph.

In any geometric representation the white
cycle is represented by a closed Jordan curve.
We will refer to it as the outer circle.
It divides the plane into two faces
-- an interior and an exterior.

The outer circle cannot be crossed by the representative of any
black vertex. Moreover, as two black vertices are adjacent and therefore their
representatives intersect each other, they have
to be represented in the same face (with respect to the outer circle).
Therefore, along this circle the representatives of gray vertices appear in
a prescribed ordering. This implies the ordering in which some
part of representatives of the black vertices occur.

The vertices $a$ and $b$ are represented as two mutually intersecting segments.
{Note that gray neighbors of $a$ and $b$ appear alternately (with respect to the outer circle) and partition the face containing $a$ and $b$ among the four quadrants. The segment representing $c$ must cross segments representing $a$ and $b$. The vertex $c$ has four gray neighbors (thus its segment must intersect their segments) and each of them has to appear in a different quadrant defined by the gray neighbors of $a$ and $b$ (as they have no neighbors except for $c$ and the vertices on the outer circle). However, this is impossible as a single segment may meet at most three quadrants.
}

Thus $G$ is not a SEG graph. On the other hand, $G$ is a $2\DIRCONV$-graph (see Figure \ref{counterexample2}) and therefore a $k\DIRCONV$ graph for any $k \geq 2$.
\end{proof}

\begin{figure}[ht]
\centering
\begin{center}
\tiny
\begin{tikzpicture}[scale=0.7]
\tikzstyle{every node}=[draw, shape = circle]

\node (o1) at (0,3) {};
\node (o2) at (1,2.6) {};
\node (o3) at (2,2) {};
\node (o4) at (2.6,1) {};
\node (o5) at (3,0) {};
\node (o6) at (2.6,-1) {};
\node (o7) at (2,-2) {};
\node (o8) at (1,-2.6) {};
\node (o9) at (0,-3) {};
\node (o10) at (-1,-2.6) {};
\node (o11) at (-2,-2) {};
\node (o12) at (-2.6,-1) {};
\node (o13) at (-3,0) {};
\node (o14) at (-2.6,1) {};
\node (o15) at (-2,2) {};
\node (o16) at (-1,2.6) {};

\node[fill = black] (a) at (0,0) {\color{white}{c}}; 
\node[fill = black] (b) at (-1,0) {\color{white}{b}}; 
\node[fill = black] (c) at (0,1) {\color{white}{a}}; 

\node[fill = gray] (a3) at (1.2,1.2) { };
\node[fill = gray] (a7) at (1.2,-1.2) { };
\node[fill = gray] (a11) at (-1.2,-1.2) { };
\node[fill = gray] (a15) at (-1.2,1.2) { };
\node[fill = gray] (b1) at (-0.5,1.5) { };
\node[fill = gray] (b9) at (-0.5,-1.5) { };
\node[fill = gray] (c5) at (1.5,0.5) { };
\node[fill = gray] (c13) at (-1.5,0.5) { };
\draw (o1)--(o2)--(o3)--(o4)--(o5)--(o6)--(o7)--(o8)--(o9)--(o10)--(o11)--(o12)--(o13)--(o14)--(o15)--(o16)--(o1);
\draw (a) -- (a3) -- (o3);
\draw (a) -- (a7) -- (o7);
\draw (a) -- (a11) -- (o11);
\draw (a) -- (a15) -- (o15);
\draw (b) -- (b1) -- (o1);
\draw (b) -- (b9) -- (o9);
\draw (c) -- (c5) -- (o5);
\draw (c) -- (c13) -- (o13);
\draw (a) -- (b) -- (c) -- (a);
\end{tikzpicture}
\hskip 10 pt
\begin{tikzpicture}[scale=0.4]
\draw (-2,-2) -- (2,-2) -- (2,2) -- (-2,2) -- cycle;
\draw (-1,-3) -- (1, -3) -- (1, 3) -- (-1, 3) -- cycle;
\draw (-3,-1) -- (3,-1) -- (3,1) -- (-3,1) -- cycle;

\draw (-6,4) --++ (4.5,0) --++ (0,1) --++ (-4.5,0) -- cycle;
\draw (-1.5,6) --++ (3,0) --++ (0,1) --++ (-3,0) -- cycle ;
\draw (1.5,4) --++ (4.5,0) --++ (0,1) --++ (-4.5,0) -- cycle;
\draw (6,-2) --++ (1,0) --++ (0,4) --++ (-1,0) -- cycle;
\draw (-1.5,-7) --++ (3,0) --++ (0,1) --++ (-3,0) -- cycle ;
\draw (1.5,-5) --++ (4.5,0) --++ (0,1) --++ (-4.5,0) -- cycle;
\draw (-7,-2) --++ (1,0) --++ (0,4) --++ (-1,0) -- cycle;
\draw (-6,-5) --++ (4.5,0) --++ (0,1) --++ (-4.5,0) -- cycle;

\draw (1,4.5) --++ (1,0) --++ (0,2) --++ (-1,0) -- cycle;
\draw (5.5,1.5) --++ (1,0) --++ (0,3) --++ (-1,0) -- cycle;
\draw (5.5,-4.5) --++ (1,0) --++ (0,3) --++ (-1,0) -- cycle;
\draw (1,-6.5) --++ (1,0) --++ (0,2) --++ (-1,0) -- cycle;
\draw (-2,-6.5) --++ (1,0) --++ (0,2) --++ (-1,0) -- cycle;
\draw (-6.5,-4.5) --++ (1,0) --++ (0,3) --++ (-1,0) -- cycle;
\draw (-6.5,1.5) --++ (1,0) --++ (0,3) --++ (-1,0) -- cycle;
\draw (-2,4.5) --++ (1,0) --++ (0,2) --++ (-1,0) -- cycle;

\draw (-0.1, 2.75) --++ (0.2,0) --++ (0, 3.5) --++ (-0.2,0) -- cycle;
\draw (-1.8,1.75) --++ (0.2,0) --++ (0,2.5) --++ (-0.2,0) -- cycle;
\draw (2.75, -0.1) --++ (3.5,0) --++ (0, 0.2) --++ (-3.5,0) -- cycle;
\draw (1.8,1.75) --++ (-0.2,0) --++ (0,2.5) --++ (0.2,0) -- cycle;
\draw (-0.1, -2.75) --++ (0.2,0) --++ (0, -3.5) --++ (-0.2,0) -- cycle;
\draw (-1.8,-1.75) --++ (0.2,0) --++ (0,-2.5) --++ (-0.2,0) -- cycle;
\draw (-2.75, -0.1) --++ (-3.5,0) --++ (0, 0.2) --++ (3.5,0) -- cycle;
\draw (1.8,-1.75) --++ (-0.2,0) --++ (0,-2.5) --++ (0.2,0) -- cycle;
\end{tikzpicture}
\end{center}
\caption{A graph $G$ and its representation as a $2\DIRCONV$ graph.}
\label{counterexample2}
\end{figure}
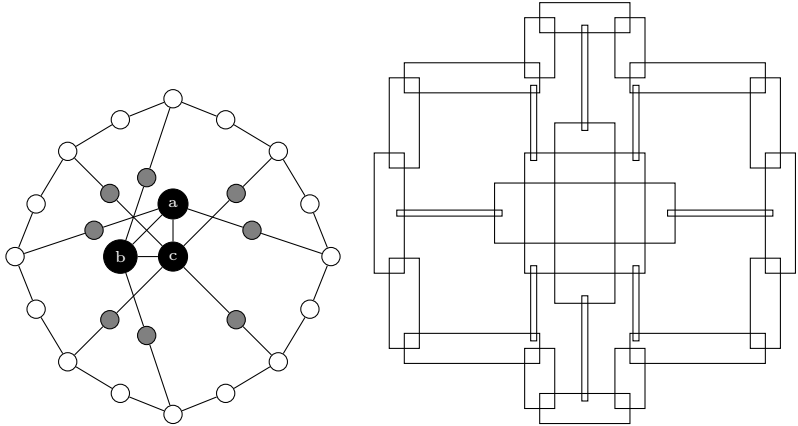

\sloppypar We conclude this section by exhibiting the relationship of $P_{hom}$ and $k\DIRCONV$ graphs.
It is clear that for every $k$ and a polygon $P \in \P(k)$, we have $P_{hom} \subseteq k\DIRCONV$. The corollary below shows that there are graphs which are in $2\DIRCONV$, but not in $P_{hom}$ for any convex polygon $P$.

\begin{proposition} \label{lem:prelphom}
$k\DIRCONV \not\subseteq P_{hom}$ for any convex polygon $P$ and  $k\geq2$.
\end{proposition}
\begin{proof}
Let us again consider the graph $G$ depicted in Figure \ref{counterexample2}
and a similar argument to what we already used.
Any intersection representation of this graph by convex polygons requires
two polygons whose boundaries intersect at least four times. As we can see in Figure \ref{counterexample2}, this is not a problem for $2\DIRCONV$ graphs, 
but our Lemma \ref{lemma:homarepdisk} shows that any representation
by homothetic convex polygons forms an arrangement of pseudodisks.
Since the boundaries of two pseudodisks may intersect at most twice, our
graph is not a $P_{hom}$ graph for any $P$.
\end{proof}

\section{The number of maximal cliques}

In this section we are interested in upper and lower bounds for the maximum number of maximal cliques in $k\DIRCONV$ and $P_{hom}$ graphs. {
The maximal clique problem finds numerous applications in telecommunication and social network analysis, error-correcting codes design, 
fault diagnosis on large multiprocessor systems, computer vision, and pattern recognition. See \cite{bomze} for a detailed presentation of the above, 
\cite{Brimkov}   for more references and discussion on applications of the maximal clique problem on geometric graphs to imaging sciences, 
and \cite{SODA,tomita,rhodes} and the bibliography therein for applications to research in bioinformatics and computational chemistry.
 The problem of enumerating all maximal cliques and finding the maximum clique is discussed in Section \ref{sec:algorithms}.

\subsection{Upper bound}
The following theorem is the main result of this section.
\begin{theorem}  \label{thm:max-dirconv}
Let $G$ be a $k\DIRCONV$ graph with $n$ vertices. Then $G$ has at most $n^k$ maximal cliques. 
\end{theorem}

\begin{proof}
Fix a representation of $G = (V,E)$ by intersecting polygons, whose every side is parallel to one of the lines $\{\ell_1,\ldots \ell_k\}$. For every $i\in \{1,\ldots,k\}$ let $w_i$ be an arbitrary vector perpendicular to $\ell_i$ (so for each line we can choose the direction of $w_i$ in two ways).

 Let $M$ be a maximal clique in $G$. For $i \in \{1,\ldots,k\}$, let $R_i$ be the last polygon in $M$ found by a sweeping line in the direction of $w_i$ (if there is more than one  such a polygon, we choose an arbitrary one). Let $r_i$ be the supporting line of $R_i$ perpendicular to $w_i$ (which is exactly the sweeping line in direction of $w_i$ that found $R_i$). 
Define $\widetilde{M}=\{R\in V: \forall i\in \{1,\ldots,k\} \; R \text{ intersects } r_i\}$. We shall prove that $M = \widetilde M$.

First let us show that $M\subseteq \widetilde{M}$. Let $R$ be a polygon in $M$ and consider some $i\in \{1,\ldots,k\}$. The polygon $R$ cannot be farther than $r_i$ in the direction $w_i$, because $r_i$ is the supporting line of $R_i$, which is the last polygon in $M$ in direction of $w_i$. On the other hand, the polygon $R$ cannot lie entirely before $r_i$ (again in the direction of $w_i$), because then $r_i$ would separate two polygons $R$ and $R_i$ in $M$. Hence $R$ intersects $r_i$ for every $i\in \{1,\ldots,k\}$ and thus it belongs to $\widetilde{M}$. 

Now we will show that $\widetilde{M}$ is a clique, which, combined with the fact that $\widetilde{M}$ contains the maximal clique $M$, implies that $\widetilde{M}=M$. 
Suppose there exist polygons $Q,R\in \widetilde{M}$ such that $Q\cap R=\emptyset$. {Since $Q$ and $R$ are convex, by Lemma \ref{lem:separate2} they} can be separated by a line $\ell$, which is parallel to one of the faces of $Q$ or $R$ and thus to some $\ell_j \in \{\ell_1,\ldots,\ell_k\}$.
Since $Q,R \in \widetilde M$, we know that both $Q$ and $R$ intersect $r_j$. Hence we get two parallel lines $\ell$ and $r_j$, one separating $Q$ and $R$ and the other intersecting both of them. It is easy to verify that this is not possible, so $\widetilde M$ is a maximal clique.

Note that each maximal clique is uniquely determined by the choice of $r_1,r_2,\ldots,r_k$. For each $i$, every $r_i$ is determined by $R_i$, which can be chosen in at most $n$ ways. So finally, the number of maximal cliques in $G$ is at most $n^k$.
\end{proof}

Theorem~\ref{thm:max-dirconv} implies the following corollary.

\begin{corollary}
Let $P$ be a convex polygon in $\P(k)$. Then any $n$-vertex graph in $P_{hom}$ has at most $n^k$ maximal cliques.
\end{corollary}

\subsection{Lower bounds}
In this section we obtain lower bounds for the maximum number of maximal cliques in $k\DIRCONV$ graphs and the subclasses of this class.
First we focus on $k$-DIR graphs (recall from Theorem \ref{kdir-kdirconv} that $k$-DIR $\subseteq k\DIRCONV$).

\begin{theorem} \label{thm:kdir-lower}
For any $k \geq 2$, the maximum number of maximal cliques over all $n$-vertex graphs in $k$-DIR is $\Omega(n^{k(1-\epsilon)})$,  for any $\epsilon > 0$.
Moreover, if $k$ is a constant, then the bound is $\Omega(n^{k})$.
\end{theorem}
\begin{proof}
Let $n$ be divisible by $k$. Let $ \{\ell_1,..,\ell_k\}$ be a set of pairwise non-parallel lines. For $i\in \{1,.., k\}$ let $S_i$ be a segment parallel to $\ell_i$. We make $\frac{n}{k}$ copies of $S_i$ for every $i\in\{1,.., k\}$. Let $S_{i,s}$ be the $s$-th copy of $S_i$. We place all segments $S_{i,s}$ for $i\in\{1,.., k\}$, $s\in\{1,.., \frac{n}{k}\}$, in such a way that $S_{i,s}$ and $S_{j,t}$ ($i,j\in\{1,.., k\}$, $s,t\in\{1,.., \frac{n}{k}\}$) intersect if and only if $i\neq j$ (see Figure \ref{lower-kdir}).

Let $G(n,k)$ be the intersection graph of $\{S_{i,s}: i\in\{1,.., k\}, s\in\{1,.., \frac{n}{k}\} \}$. Notice that $G(n,k)$ is a complete $k$-partite graph $K_{\frac{n}{k},\ldots, \frac{n}{k}}$ and every maximal clique contains exactly one of $S_{i,1},.., S_{i,\frac{n}{k}}$ for every $i\in \{1,..,k\}$. Hence the number of maximal cliques in $G$ is $(\frac{n}{k})^k=\Omega(n^{k-\log k })=\Omega(n^{k(1-\epsilon)})$ for any $\epsilon > 0$.

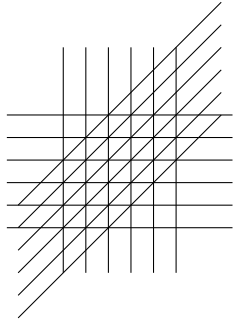
\begin{figure}[ht]
\begin{center}
\begin{tikzpicture}[xscale=0.3, yscale=0.3]
\newcommand*\nn{6}
\newcommand*\rea{--++(0,10)}
\newcommand*\reb{--++(10,0)}
\newcommand*\rec{--++(9,9)}

 \foreach \x in {1,...,\nn}
{
    \draw[color = black] (0.0 + 1* \x, 0.0) \rea;
    \draw[color = black] (-1.5, 1.0 + 1 * \x) \reb;
    \draw[color = black] (-1.0, -3.0 + 1 * \x) \rec;
}
\end{tikzpicture}
\end{center}
\caption{Construction for $k=3$ and $n=18$.}
\label{lower-kdir}
\end{figure}
\end{proof}

We can also obtain the same bound as in Theorem \ref{thm:kdir-lower}, but for $P_{hom}$ graphs, where $P$ is a $2k$-gon in $\P(k)$, i.e., $P$ has $k$ pairs of parallel sides.

\begin{theorem} \label{thm:2kphom-lower}
For any $2k$-gon $P \in \P(k)$ there exists an $n$-vertex graph $G\in \phom$ with $\Omega(n^{k(1-\epsilon)})$ maximal cliques,  for any $\epsilon > 0$.
Moreover, if $k$ is a constant, then the bound is $\Omega(n^{k})$.
\end{theorem}

\begin{proof}
Let $L = \{\ell_1,\ell_2,\ldots,\ell_k\}$ be the directions of sides of $P$, i.e., $P \in \P(L)$.
For a copy $P'$ of $P$, by $s_{1,d}(P')$ and $s_{2,d}(P')$ we shall denote the two sides of $P'$, which are parallel to $\ell_d$.

For $d\in \{1,2,..,k\}$ and $i\in \{1,\ldots, \frac{n}{k}\}$ let $L_{d,i}$ and $R_{d,i}$ be a pair of two disjoint copies of $P$, such that sides $s_{1,d}(L_{d,i})$ and $s_{2,d}(R_{d,i})$ are parallel and very close to each other in the sense that one can be obtained from the other by a translation by a short vector perpendicular to $\ell_d$.

Pairs $(L_{d,1},R_{d,1}), \ldots , (L_{d,\frac{n}{k}},R_{d,\frac{n}{k}})$ are placed in such a way that the following conditions are satisfied (see Figure \ref{lower6}, left):
\begin{itemize}
\item $L_{d,i}$ intersects $R_{d,j}$ for all $i,j\in \{2,\ldots \frac{n}{k}\}$ such that $j < i$,
\item $L_{d,i}$ intersects $L_{d,j}$ and $R_{d,i}$ intersects $R_{d,j}$ for all $i,j\in \{2,\ldots \frac{n}{k}\}$.
\end{itemize}

Observe that in the intersection graph, the vertices $L_{d,1},\ldots L_{d,\frac{n}{k}}$ (and also $R_{d,1},\ldots R_{d,\frac{n}{k}}$) form a clique.

Let $G_d$ denote the subgraph induced by $\{L_{d,i},R_{d,i} \colon i\in \{1,\ldots ,\frac{n}{k}\}\}$. In $G_d$ all maximal cliques are of the form: $\{L_{d,i},L_{d,i+1},\ldots,L_{d,\frac{n}{k}},R_{d,1},\ldots R_{d,i-1} \}$, so there are $\frac{n}{k}$ maximal cliques in $G_d$ for every $d\in \{1,2,..,k\}$.

Notice that polygons $L_{d,i}$, $R_{d,i}$ can be placed in such a way that $L_{d_1,i}$ and $L_{d_2,j}$ (and, by symmetry  $R_{d_1,i}$ and $R_{d_2,j}$) intersect each other for all distinct  $d_1,d_2\in \{1,,..,k\}$ and all $i,j\in \{1,\ldots,\frac{n}{k}\}$ (see Figure \ref{lower6}, right).

Hence every maximal clique in $G$ is a disjoint union of $k$ cliques, each in some $G_d$ for $d\in \{1,2,..,k\}$. Therefore, the number of maximal cliques is $(\frac{n}{k})^k=\Omega(n^{k-\log k })=\Omega(n^{k(1-\epsilon)})$ for any $\epsilon > 0$.
\end{proof}

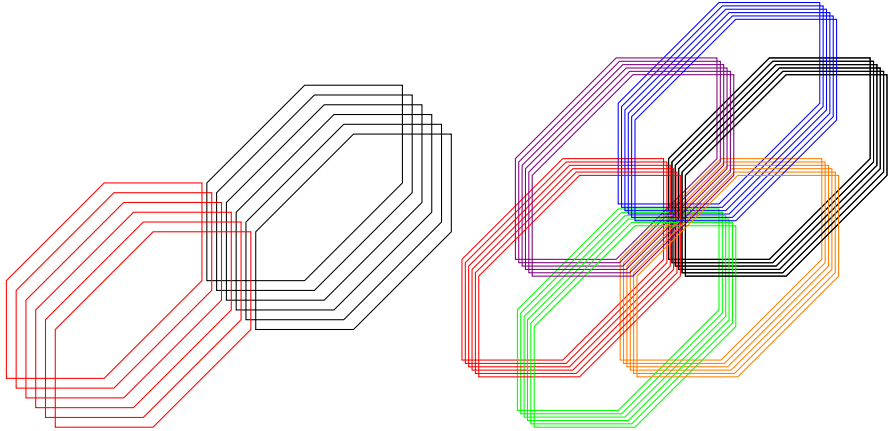
\begin{figure}[ht]
\begin{center}
\begin{tikzpicture}[scale=1.3]
\newcommand*\nn{5}
\newcommand*\x{0}
\newcommand*\f{10}

\newcommand*\pv{--++(1,0)--++(1,1)--++(0,1)--++(-1,0)--++(-1,-1)--cycle}

\foreach \x in {0,...,\nn}
 {
    \draw[color=black, scale = 1] (0 + \x/\f,0.5-\x/\f) \pv;
	\draw[color=red, scale = 1] (-2.05 + \x/\f,-0.5-\x/\f) \pv;    
 }
\end{tikzpicture}
\begin{tikzpicture}[scale=1.34]
\newcommand*\nn{5}
\newcommand*\x{0}
\newcommand*\f{30}

\newcommand*\pv{--++(1,0)--++(1,1)--++(0,1)--++(-1,0)--++(-1,-1)--cycle}

\foreach \x in {0,...,\nn}
 {
    \draw[color=black, scale = 1] (0 + \x/\f,0.5-\x/\f) \pv;
	\draw[color=red, scale = 1] (-2.05 + \x/\f,-0.5-\x/\f) \pv;    
	\draw[color=blue, scale = 1] (-0.5 + \x/\f,1.05-\x/\f) \pv;
	\draw[color=green, scale = 1] (-1.5 + \x/\f,-1-\x/\f) \pv;     
	\draw[color=black, scale = 1] (0 + \x/\f,0.5-\x/\f) \pv;
	\draw[color=orange, scale = 1] (-0.48 + \x/\f,-0.5-\x/\f) \pv;
	\draw[color=violet, scale = 1] (-1.52 + \x/\f,0.5-\x/\f) \pv;
 }

\end{tikzpicture}
\end{center}
\caption{{\em Left:} A placement of polygons $L_{d,i}$ and $R_{d,i}$ in the construction from Theorem \ref{thm:2kphom-lower}. {\em Right:} The representation of $G$.}
\label{lower6}
\end{figure}

Observe that the construction in fact works for $P_{translate}$ graphs, as we do not use any scaling. Also note that Theorem \ref{thm:2kphom-lower} gives a tight bound for $P$ being a parallelogram {(as it is a $2\DIRCONV$ graph)}. 

We can provide a very similar construction for all regular polygons, even if the number of sides is odd. However, in this case the number of maximal cliques is much lower { with respect to the number of directions of sides} (although still increasing with $k$).

\begin{theorem} \label{thm:k-reg-phom-lower}
For any $k$ and any regular $k$-gon $P$ there exists an $n$-vertex graph $G\in \phom$ with $\Omega(n^{\lfloor \frac{k}{2}\rfloor(1-\epsilon)})$ maximal cliques,  for any constant $\epsilon > 0$.
Moreover, if $k$ is a constant, then a bound $\Omega(n^{\lfloor \frac{k}{2}\rfloor})$ holds.
\end{theorem}
\begin{proof}
The case where $k$ is even is covered by Theorem \ref{thm:2kphom-lower}. For the case where $k$ is odd, a construction very similar to the proof of Theorem \ref{thm:2kphom-lower} works.

For simplicity, set $q := \lfloor k/2 \rfloor$.
Let $F_1,F_2,..,F_q$ be $q$ consecutive sides of $P$. For each $d \in \{1,2,..,q\}$, let $c_d$ denote the  corner of $P$ which lies opposite the side $F_d$. For a copy $P'$ of $P$, by $f_d(P')$ (resp., $c_d(P')$) we shall denote the appropriate side (resp., corner) of $P'$.

For every $d\in \{1,2,..,q\}$ and $i\in \{1,\ldots, \frac{n}{q}\}$ we take a pair $L_{d,i}$, $R_{d,i}$ of copies of $P$ and place them in such a way that they are disjoint, but $F_d(L_{d,i})$ is very close to $c_d(R_{d,i})$ (see Figure \ref{lower9}, left).

Pairs $(L_{d,1},R_{d,1}), \ldots , (L_{d,n/q},R_{d,n/q})$ are placed is such a way that the following conditions are satisfied (again, refer to Figure \ref{lower9}, right):
\begin{itemize}
\item $L_{d,i}$ intersects $R_{d,j}$ for all $i,j\in \{2,\ldots n/q\}$ such that $j < i$,
\item $L_{d,i}$ intersects $L_{d,j}$ and $R_{d,i}$ intersects $R_{d,j}$ for all $i,j\in \{2,\ldots n/q\}$.
\end{itemize}

By the same reasoning as in the previous proof one can verify that the intersection graph for this configuration of polygons has $\Omega \left( (n/q)^q \right)$ maximal cliques.
\end{proof}

\begin{figure}[ht]
\begin{center}
\begin{tikzpicture}[scale=0.65]
\newcommand*\m{4}
\newcommand*\f{5}
\newcommand*\kk{9}

\foreach \x in {0,...,\m}
{
\node[draw,minimum size=2.5cm,regular polygon,regular polygon sides=\kk, color = blue] at (1.95+\x/\f,1.63+\x/\f/2) {};
\node[draw,minimum size=2.5cm,regular polygon,regular polygon sides=\kk, color = green] at (-1.77+\x/\f,2.5+\x/\f/2) {};
}
\end{tikzpicture}
\begin{tikzpicture}[scale=1.3*3/5]
\newcommand*\m{2}
\newcommand*\f{30}
\newcommand*\kk{9}

\foreach \x in {0,...,\m}
{
\node[draw,minimum size=3cm,regular polygon,regular polygon sides=\kk] at (0+\x/\f,0+\x/\f) {};
\node[draw,minimum size=3cm,regular polygon,regular polygon sides=\kk, color = red] at (0+\x/\f,3.8+\x/\f) {};

\node[draw,minimum size=3cm,regular polygon,regular polygon sides=\kk, color = blue] at (1.95+\x/\f,1.63+\x/\f) {};
\node[draw,minimum size=3cm,regular polygon,regular polygon sides=\kk, color = green] at (-1.77+\x/\f,2.5+\x/\f) {};

\node[draw,minimum size=3cm,regular polygon,regular polygon sides=\kk, color = orange] at (1.34-\x/\f,0.5+\x/\f) {};
\node[draw,minimum size=3cm,regular polygon,regular polygon sides=\kk, color = violet] at (-1.2-\x/\f,3.3+\x/\f) {};

\node[draw,minimum size=3cm,regular polygon,regular polygon sides=\kk, color = yellow] at (1.74-\x/\f,2.95+\x/\f) {};
\node[draw,minimum size=3cm,regular polygon,regular polygon sides=\kk, color = pink] at (-1.5-\x/\f,1.0+\x/\f) {};
}
\end{tikzpicture}
\end{center}
\caption{{\em Left:} A placement of polygons $L_{d,i}$ and $R_{d,i}$ in the construction from Theorem \ref{thm:k-reg-phom-lower}. {\em Right:} The representation of $G$.}
\label{lower9}
\end{figure}
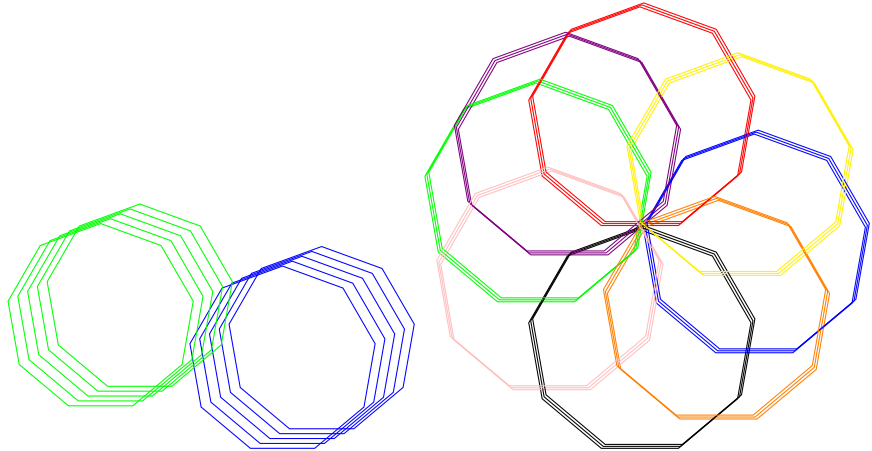

Again, observe that the construction above in fact works for $P_{translate}$ graphs.
We strongly believe that a similar construction can be conducted for any polygon $P$ with at least 4 directions of sides.

\begin{conjecture}
For every $k \geq 4$ and every convex $k$-gon $P$
there exists an infinite family of $P_{hom}$ graphs with $\Omega(n^{\lfloor k/2 \rfloor(1-\epsilon)})$ maximal cliques, for any $\epsilon > 0$.
\end{conjecture}

{
We also believe that there are polygons, for which this bound is asymptotically tight.

\begin{conjecture}
For every sufficiently large $k$ there exists a convex $k$-gon $P$
such that every $P_{hom}$ graph has at most $O(n^{ k/2 })$ maximal cliques.
\end{conjecture}
}

As a last result in this section, we give a general bound for $P_{hom}$ graphs, where $P$ is any convex polygon but a parallelogram. It is a simplified and generalized version of the construction for homothetic triangles, presented by Kaufmann {\em et al.} \cite{SODA}.

\begin{theorem}  \label{lowern3}
If $P$ is not a parallelogram then the maximum number of maximal cliques in an $n$-vertex graph in $P_{hom}$ is $\Omega(n^3)$.
\end{theorem}

\begin{proof}
Let $F$ be a face of $P$ and $\ell$ be the line containing $F$. Consider the set of vertices of $P$, which are at the largest distance from $\ell$. If there is only one such vertex, denote it by $D(F)$. It there are two such vertices let $D(F)$ denote the face spanned by these two vertices.

Choose a side of $P$ and call it $F_1$ and let $P_1 = D(F_1)$.
Let $F_2,F_3$ be sides of $P$ adjacent to $P_1$ and let $P_2 = D(F_2)$ and $P_3 = D(F_3)$.

Let $h,r,t,v$ be four copies of $P$. By $F_i^h,F_i^r,F_i^t,F_i^v,P_i^h,P_i^r,P_i^t,P_i^v$ we denote the sides and corners in polygons $h,r,t,v$ corresponding to $F_i,P_i$ in polygon $P$ for $i\in \{1,2,3\}$, respectively (see Figure \ref{n3}).
We can adjust the sizes and positions of $h,r,t,v$ in such a way that:
\begin{enumerate}
\item $t$ and $h$ are touching and $F_1^t$ intersects $P_1^h$,
\item $t$ and $v$ are touching and $F_1^t$ intersects $P_1^v$,
\item $h$ and $v$ are touching and $F_2^h$ intersects $P_2^v$,
\item $r$ and $t$ are touching and $F_3^r$ intersects $P_3^t$,
\item $r$ and $h$ intersect.
\item $r$ and $v$ intersect.
\end{enumerate}

For every polygon $x\in \{h,r,t,v\}$ we make $\frac{n}{4}$ copies $x_1,..x_\frac{n}{4}$ and move them slightly with respect to the position of $h,r,t,v$ in such a way that:
\begin{enumerate}
\item $t_i$ and $v_j$ intersect iff $ i\ge j$,
\item $h_i$ and $v_j$ intersect iff $i\le j $,
\item $ h_j$ and $t_j$ intersect iff $ i\le j$,
\item $ r_i$ and $t_j$ intersect iff $i\ge j$,
\item $r_i$ and $v_j$ intersect  for all $i,j\in\{1,..,\frac{n}{4}\}$,
\item $ r_i$ and $h_j$ intersect for all $i,j\in\{1,..,\frac{n}{4}\}$.
\end{enumerate}
\sloppypar For any $\alpha,\beta,\gamma$ such that $1\le \alpha\le \beta\le \gamma\le \frac{n}{4}$ the set $\{h_1,.., h_\alpha,v_\alpha,.., v_\beta,t_\beta,.., t_\gamma,r_\gamma,.., r_\frac{n}{4}\}$ is a maximal clique in $G$. Hence there are at least ${\frac{n}{4} \choose 3}=\Omega(n^3)$ maximal cliques in total.

\begin{figure}[ht]
\begin{center}
\begin{tikzpicture}[scale=0.25]
\newcommand*\nn{4}
\newcommand*\sv{1}
\newcommand*\sh{3}
\newcommand*\st{5}
\newcommand*\sr{5}
\newcommand*\pv{--++(\sv * 3,\sv* 0)--++(\sv * 2,\sv * 2)--++(\sv * -3,\sv * 3)--++(\sv * -3,\sv * -2)--cycle}
\newcommand*\ph{--++(\sh * 3,\sh* 0)--++(\sh * 2,\sh * 2)--++(\sh * -3,\sh * 3)--++(\sh * -3,\sh * -2)--cycle}
\newcommand*\pt{--++(\st * 3,\st* 0)--++(\st * 2,\st * 2)--++(\st * -3,\st * 3)--++(\st * -3,\st * -2)--cycle}
\newcommand*\pr{--++(\sr * 3,\sr* 0)--++(\sr * 2,\sr * 2)--++(\sr * -3,\sr * 3)--++(\sr * -3,\sr * -2)--cycle}

\foreach \x in {0,...,\nn}
 {
    \draw[color=black, scale = 1] (-20 - \x/2,-15 + \x) \pt; 
    \draw[color=blue] (0 + \x /4,0 + \x /4) \pv;                   
    \draw[color=green] (-5 + \x/4,5 + \x/4) \ph; 
    \draw[color=red] (4.5 + \x /4,-6 - \x /4) \pr;      
 }

\end{tikzpicture}
\end{center}
\caption{Construction for the lower bound of $\Omega(n^3)$.}
\label{n3-full}
\end{figure}
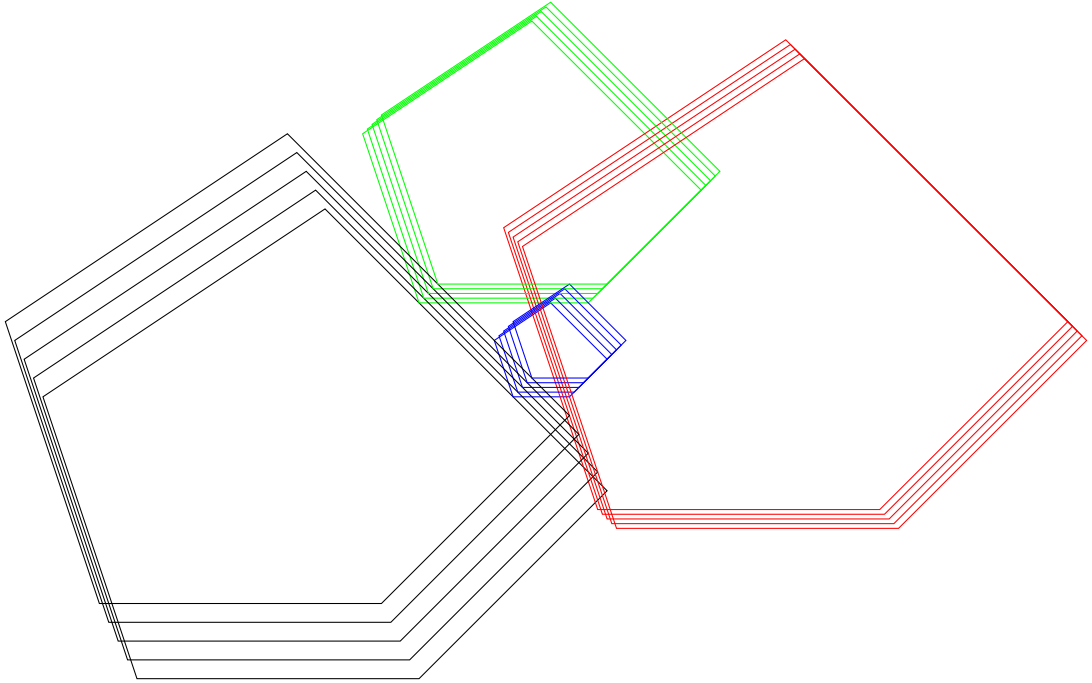

\begin{figure}[ht]
\begin{center}
\begin{tikzpicture}[scale=0.5]
\draw[line width = 1.5] (0,0)--(3,0);
\draw (3,0)--(5,2);
\draw[line width = 1.5] (5,2)--(2,5);
\draw (2,5)--(-1,3);
\draw[line width = 1.5] (-1,3)--(0,0);

\node at (4, 4) {\scriptsize $F_1$};
\node at (1.5, -0.5) {\scriptsize  $F_2$};
\node at (-1, 1.5) {\scriptsize  $F_3$};
\node at (-0.3,-0.3) {\scriptsize  $P_1$};
\node at (2,5.5) {\scriptsize  $P_2$};
\node at (5.5,2) {\scriptsize  $P_3$};

\end{tikzpicture}
\begin{tikzpicture}[scale=0.44]
\newcommand*\nn{5}
\newcommand*\sv{1}
\newcommand*\sh{3}
\newcommand*\st{5}
\newcommand*\sr{5}
\newcommand*\pv{--++(\sv * 3,\sv* 0)--++(\sv * 2,\sv * 2)--++(\sv * -3,\sv * 3)--++(\sv * -3,\sv * -2)--cycle}
\newcommand*\ph{--++(\sh * 1 * 0.4, \sh * -3 * 0.4)--++(\sh * 3,\sh* 0)--++(\sh * 1,\sh * 1)}
\newcommand*\pt{--++(\st * 0.2,\st * 0.2)--++(\st * -2.2,\st * 2.2)}
\newcommand*\pr{--++(0.9 * \sr * -1 ,0.9 * \sr * 3)--++(\sr * 0.1 * 3 + 5/4 - \x/4,\sr * 0.1 * 2+ 5/4 - \x/4)}

\foreach \x in {1,...,\nn}
 {
    \draw[color=black] (-19.5 + \st * 4.8-\x/2,-15.5 + 1.8 * \st + \x) \pt; 
    \draw[color=blue] (0 + \x /4,0 + \x /4) \pv;                   
    \draw[color=green] (0.4 * -1 * \sh -5 + \x/4,5 + 0.4 * 3 * \sh + \x/4) \ph; 
    \draw[color=red] (4.66 + \x /4 - 0.1 * \sr * 1,-6 - \x /4 + 0.1 * \sr * 3) \pr;      

    \node at (-20 + \st * 2.8-\x/2,-15.5 + 4.2 * \st + \x) {\scriptsize   $t_\x$};
    \node at (0.4 * -1 * \sh -5.3 + 1.3 * \x/4,5.5 + 0.4 * 3 * \sh + 1.3 * \x/4) {\scriptsize   $h_\x$};
    \node at (2.9, 11.3 - 3 + \x/2) {\scriptsize $r_\x$};
    \node at (5.3 + 1.2 * \x /4,1.6 + 1.2 * \x /4) {\scriptsize  $v_\x$};
 }
\end{tikzpicture}
\end{center}
\caption{Idea of the construction (for $n=20$).}
\label{n3}
\end{figure}
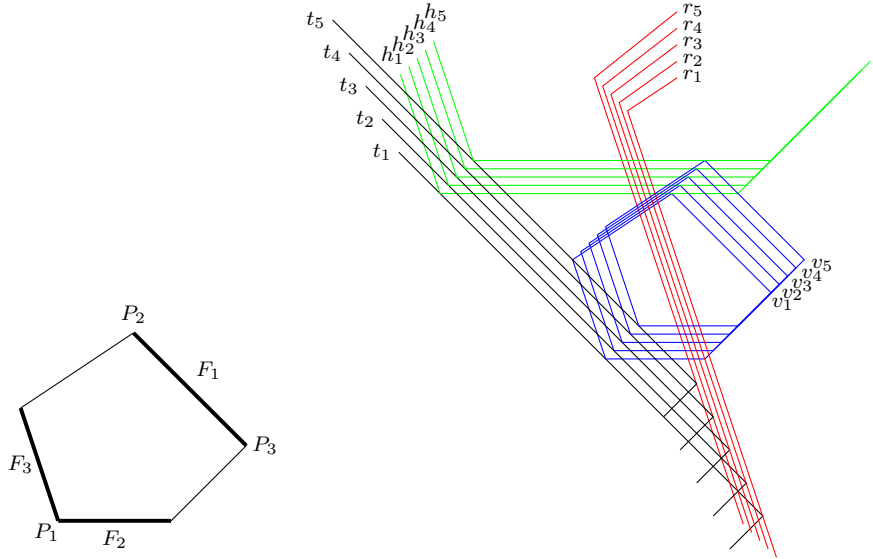
\end{proof}

\subsection{Towards higher dimensions}

In this section we generalize the concept of intersection graphs of convex polygons to arbitrary number of dimensions. The definitions we use here are straightforward generalizations of the definitions for the 2-dimensional case.

For a polytope $P$, let $\dim(P)$ denote its dimension.
Let $\L^d$ be the set of all $(d-1)$-dimensional hyperplanes in $\rr^d$, containing the origin point. For any $L \in \binom{\L^d}{k}$, by $\P^d(L)$ we denote the set of all polytopes in $\rr^d$, whose every facet (i.e., a $(d-1)$-dimensional face) is parallel to one of the hyperplanes in $L$. By $\P^d(k)$ we denote the set $\bigcup_{L \in \binom{\L^d}{k}} \P^d(L)$.

By $k\DIRCONVkd{L}$ we define the class of intersection graphs of polygons in $\P^d(L)$, while $k\DIRCONVd$ is defined as $\bigcup_{L \in \binom{\L^d}{k}} k\DIRCONVkd{L}$.

Now let us present the separation theorem for polytopes by Wright \cite{Wright}.

\begin{theorem}[\cite{Wright}]
Consider non-empty convex polytopes $P_1$ and $P_2$ in a Euclidean space and suppose that $P_1$ and $P_2$ can be properly separated\footnote{We say that a hyperplane $H$ {\em properly separates} convex sets $A$ and $B$ if at least one of those sets does not lie entirely within $H$.}. Then there exist parallel hyperplanes $H_1$ and $H_2$ properly separating $P_1$ and $P_2$, for which $\dim (H_1 \cap P_1) + \dim (H_2 \cap P_2) \geq \dim (P_1 \cup P_2) -1$.
\end{theorem}

From this theorem we can easily obtain the following corollary, generalizing Lemma \ref{lem:separate2}.

\begin{corollary} \label{cor:separate-d-dim}
Let $P_1$ and $P_2$ be disjoint convex $d$-dimensional polytopes.
Then they can be separated by a hyperplane $H$, which is parallel to some $d_1$-dimensional face of $P_1$ and to some $d_2$-dimensional face of $P_2$, such that $d_1 + d_2 = d-1$ (either $d_1$ or $d_2$ may be equal to 0).
\end{corollary}

We shall bound the maximum number of maximal cliques in a similar way as we did in the proof of Theorem \ref{thm:max-dirconv}. 

{
\begin{theorem}
Every $n$-vertex graph in $k\DIRCONVd$ has at most $n^{d \cdot k^{d+1}}$ maximal cliques.
\end{theorem}

\begin{proof}
Let $G$ be an $n$-vertex graph in $k\DIRCONVd$.
Fix some representation of $G$ with $d$-dimensional polytopes, whose every facet is parallel to one hyperplane from  $L = \{\ell_1,..,\ell_k\} \in \binom{\L^d}{k}$.

Let $P_1, P_2 \in \P^d(L)$. Let $\mathcal{H}$ be the set of all $(d-1)$-dimensional hyperplanes $H$ containing the origin point, such that $H$ is parallel to some $i$-dimensional face of $P_1$ and to some $(d-1-i)$-dimensional face of $P_2$ (for $i \in \{0,1,..,d-1-i\}$). Define $h:=|\mathcal{H}|$ and let $\mathcal{H} = \{H_1,H_2,..,H_h\}$.
Notice that each $i$-dimensional face is defined as an intersection of some $d-i$ facets. Thus, we have 
$$
h \leq \sum_{i=0}^{d-1} {k \choose d-i} {k \choose i+1} \leq \sum_{i=0}^{d-1} k^{d-i} k^{i+1} = d \cdot k^{d+1}.
$$
For every $j\in \{1,\ldots,h\}$, let $w_j$ be an arbitrary normal vector of $H_j$. Now we can proceed in the same way as in the proof of Theorem \ref{thm:max-dirconv}, considering sweeping $(d-1)$-dimensional hyperplanes in each direction $w_j$ (for $j \in \{1,2,..,h\}$), instead of sweeping lines.

Thus, it follows that number of maximal cliques in $G$ is at most $n^h \leq n^{d \cdot k^{d+1}}$.
\end{proof}
}

\subsection{Parametrized complexity of the {\sc Clique} problem in $P_{hom}$ graphs} \label{sec:algorithms}
In this paper we have shown that the number of maximal cliques in any $k\DIRCONV$ graph (and therefore in any $P_{hom}$ graph for $P \in \P(k)$) is at most $n^{k}$. 
Tsukiyama {\em et al.} \cite{maxcliques} presented an algorithm enumerating all the maximal cliques in an $n$-vertex graph in time $O(n^3 \cdot C)$, where $C$ is the number of maximal cliques. Thus the {\sc Clique} problem  can be solved in time $O(n^{k+3})$ for any $G \in k\DIRCONV$ (even if the geometric representation is not known), and therefore is in XP when parameterized by $k$. 
One can observe that the proof of Theorem \ref{thm:max-dirconv} yields a slightly better $O(k \cdot n^{k+2})$ algorithm for this problem. However, it requires that the geometric representation of the input graph is given.

It is interesting to know if the problem is in FPT, or more generally, to answer the following question.

\begin{problem}
What is the parametrized complexity of the {\sc Clique} problem for $P_{hom}$ graphs (parametrized by the number of directions of sides of $P$)?
\end{problem}
\vskip 5 pt

\section{Recognition}

In this section we prove some results concerning the hardness of recognition of $k\DIRCONV$ and $P_{hom}$ graphs. Before we start showing our results, let us present relevant known facts that we use later. We are using three main tools.

Our results are using reduction of Hlin\v en\'y and Kratochv\'\i l \cite{HK}. We show that this reduction works also for classes we are interested in. It reduces the {\sc E3-Nae-Sat}(4) problem to disk- and pseudodisk-graph recognition.

The instance of {\sc Nae-Sat} is a boolean formula in conjunctive
normal form and we ask whether there exists a (satisfying) assignment
such that in no clause all the literals are evaluated true.
The version {\sc E3-Nae-Sat}(4) is a restriction of {\sc Nae-Sat}
to formulae with each clause consisting of {\bf exactly} three literals
and each variable occurs at most four times. For a given formula, the
reduction starts by taking an incidence graph of this formula. Incidence
graph is a bipartite graph whose one part is formed by variables, the other
part is fromed by clauses and an edge indicates that a particular variable
belongs to a given clause. 

The reduction replaces individual parts of the incidence graph by
individual gadgets: Vertices corresponding to variables are replaced
by variable gadgets, vertices corresponding to clauses are replaced
by clause gadgets, edges are replaced by incidence gadgets. The
reduction proceeds w.r.t. particular planar embedding of the incidence
graph, each mutual intersection of edges (in that embedding) has to be
reflected by "cross-over" gadget. The use of these cross-over gadgets is
rather tricky as will be seen later from their formal description.
Individual gadgets are designed in such
a way that the resulting graph has disk- or pseudo-disk representation
iff the original formula has the appropriate (variable-)assignment.

The second tool is very similar to the first one and can be found in \cite{KratDIR}. This time, Kratochv\'\i l used a bit different version of SAT: Planar 3-CON-3-SAT(4). Here we are asking about satisfying assignment of a formula (in conjunctive form). Individual clauses consist of 3 literals, each variable occurs at most four times and yet the incidence graph of this formula is planar and 3-connected. This problem is also shown to be NP-hard in \cite{KratDIR}. Again, gadget-replacements are applied and the resulting graph can be represented by straight line segments using (at most) 3 directions (3-DIR) or it even does not allow a string representation (i.e., by arc connected curves in a plane). Due to this fact, any class containing 3-DIR and simultaneously being contained in class of string graphs has to be NP-hard for recognition.

The third tool we will use was proposed by Kratochv\'\i l and Matou\v sek in \cite{KratMat}. Using a clever linear programming formulation they showed that for fixed slopes the recognition of segment graphs is in NP. This tool witnesses the existence of representation of a given graph by straight-line segments in the following way: We are showing membership of the problem in NP, therefore we need to design a certificate of existence of that representation. As a certificate one takes combinatorial description of the arrangement of pseudosegments extended to pseudolines saying the order of intersections and assigning individual slopes. To witness its stretchability from the certificate, a linear program is created. It enforces individual crossings to appear in prescribed order (w.r.t. left-right orientation). The fact that a line $l: y=a_lx+b_l$ (whose parameters $a_l$ and $b_l$ we determine from the certificate) intersects line $m: y=a_mx+b_m$ before line $n: a_nx+b_n$ (w.r.t. orientation from left to right) can be expressed by the inequality 
${b_m-b_l \over a_l-a_m}<{b_n-b_l \over a_l-a_n}.$ This inequality stems from the fact that the $x$-coordinate of intersection-point of $a_l$ and $a_m$ is $x_1={b_m - b_l \over a_l - a_m}$, for $a_l$ and $a_n$ we obtain similarly $x_2$ and left-right precedence corresponds to the fact that $x_1<x_2$. As the values $b_l, b_m$ and $b_n$ are constants while $a_l, a_m$ and $a_n$ are variables (witnessing existence of "starting-points" of individual lines), these expressions give a linear program that can be solved in polynomial time, e.g., by the ellipsoid method.

Equipped by these three useful tools, we proceed to the results:

\begin{theorem}\label{thm:kdirk}
For every fixed $k \geq 2$, it is NP-complete to recognize
\begin{enumerate}
\item $k\DIRCONVk{L}$ graphs for any $L \in \binom{\L}{k}$,
\item $k\DIRCONV$ graphs.
\end{enumerate}
\end{theorem}

\begin{proof}
As $2\DIRCONV$ graphs are exactly graphs of boxicity at most 2, they are NP-complete to
recognize \cite{KratDIR}.

For $k>2$, the class of $k\DIRCONV$ graphs contains the class of 3-DIR graphs and at the same time is contained in CONV. 
Thus, to prove NP-hardness we may apply the reduction from \cite{KratDIR} mentioned as our second (third-party) tool.
Since for $k=3$ all triples of directions are equivalent under an affine transformation of the plane, this shows that for $k>2$, recognition of $k\DIRCONV$, and also of $k\DIRCONVk{L}$ or of any $k$-tuple of directions, is NP-hard.

NP-membership follows directly from \cite{KratMat} as the appropriate polygon-sides get represented by underlying lines and the problem reduces to the question whether a given pseudoline arrangement where lines have prescribed slopes, is stretchable. As the same idea is used yet in proof of the next theorem, to avoid rewriting already published facts, we mention in the next proof which its part suffices as the proof of this claim.
\end{proof}

Concerning the class of $P_{hom}$ graphs, our aim is to obtain the following result announced in \cite{KratPer}.
\begin{theorem}\label{thm:phom}
For every convex polygon $P$, the recognition problem of $P_{hom}$ graphs is NP-hard.
\label{thm:phomrecoghard}
\end{theorem}

Before proving the theorem, we observe the following.
The convex polygons are required either to properly intersect or to not intersect at all, i.e.,
they are not allowed to only touch each other at border points.
Therefore, we may consider
that any representation is perturbation resistant, i.e., we
may slightly move any polygon in any direction and obtain a
topologically equivalent configuration. 
Thus we may consider only
representations satisfying the assumptions of
Lemma~\ref{lemma:homarepdisk}.

\begin{proof}[Proof of Theorem \ref{thm:phom}]
As a first step we refer to \cite{SODA} which proves the NP-hardness for
homothetic triangles.

For polygons with more corners we use the construction introduced in \cite{HK},
i.e., first tool in our toolbox introduced at the beginning of the section.
We apply the same reduction and thus we just perform more sensitive analysis
consisting of two facts: First is that each $P_{hom}$ graph is
a pseudodisk graph. Second fact shows how to represent the appropriate
graphs as homothets of a given shape when pseudodisk representation exists
(i.e., the given formula was satisfiable).

By Lemma \ref{lemma:homarepdisk} we have that the graphs representable by the homothetic
polygons form a subset of graphs representable by pseudo-disks. Thus
it suffices to show that the graph obtained from any satisfiable
formula can be represented by homothetic polygons in a plane. 
For this, we need to introduce some terminology and denotations.

For a convex polygon $P$ we choose an orthogonal basis $b_1, b_2$
such that all sides of the polygon $P$ are not parallel to
$b_1$ or to $b_2$. Given such a basis $b_1, b_2$, we consider the
smallest axis-aligned rectangle containing $P$ (its bounding box)
and denote it by $BB(P)$. One can easily see that we may choose
even such a basis that $P$ touches $BB(P)$ inside $BB(P)$'s edges
rather than at corners.

For an arrangement of homothetic convex polygons we may pick up such
a basis $b_1, b_2$. As we will not be interested in the basis itself but
only in the bounding boxes, we will not 
require
that $BB(P)$ must
be taken with respect to a certain basis. 
The basis will be fixed in a way
to secure the condition that the corners of $BB(P)$ are not elements
of $P$.

Now we recall the reduction. As it is described in detail 
in \cite{HK} and we mentioned its idea at the beginning of this section,
now we focus on its main points:

We reduce the {\sc E3-Nae-Sat}(4) problem to representability
by homothetic polygons of any shape. The graph whose
$P_{hom}$-representability we will be questioning, consists of gadgets
for clauses, gadgets for variables, and connections between them.
This graph we obtain from the incidence-graph of a given formula
by replacing each clause-vertex by a clause gadget, variable-vertex
by a variable gadget and an edge between them by an occurence gadget.
Let us note that the incidence graph for a given formula is a bipartite
graph having a vertex for each variable and also for each clause. An edge
corresponds to the fact that the variable occurs in this clause (either
positively or negated).

The gadget for a variable
is $C_8$ (i.e., cycle with vertices $c_1,\ldots,c_8$), which can clearly
be represented
by homothetic polygons in a plane. Each occurrence of the variable is
represented by two consecutive vertices ($c_{2k-1}, c_{2k}$). If the
variable's occurrence is negated, we swap the labels of $c_{2k-1}$
and $c_{2k}$. The truth assignment in the representation
is determined by the orientation of polygons $P_{c_1},\ldots, P_{c_8}$
(they may go either clockwisely or counter-clockwisely). The
connections are represented by a ladder (see Figure~\ref{ladder})
and crucial for the construction is the fact (proved in \cite{HK})
that the ladder cannot distort. By distorting we mean swapping
``from the left to the right'', i.e., the path that started on the
left side of the other, will always be represented to the left of
the other. To be more formal, if we take stripe of plane from variable gadget to clause gadget delimited by left-most and right-most curves in the representation of ladder, the left end of this stripe will consist (only) of curves representing path that started on the left while the right end of stripe will consist (only) of curves representing path that started on the right. This follows from the fact that we are working with pseudodisks. In this way, the left-right orientation of paths in the ladder is well-defined.

\begin{figure}[h]
\begin{center}
\includegraphics{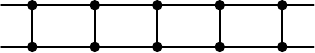}
\end{center}
\caption{The ladder}
\label{ladder}
\end{figure}

The ladders (occurence gadgets)
representing the first and the second occurrence of the particular
variable are connected by a ``cross-over''-gadget. The same applies to the third and
the fourth occurrence. We want to obtain the ladders representing
particular occurrences of the variable with respect to clockwise
orientation for positive and negative variable to appear around
the variable-gadget always in the same order (i.e., ladder 1,
ladder 2, ladder 3, and ladder 4). This is secured by implementation of the cross-over. 
If the variable is assigned a value ``true'', we just make the respective
ladders touch; if the variable is assigned a value ``false'', we cross them.
A crucial fact (proved in \cite{HK}) is that twisting a ladder does not occur 
(as there is always ``the left'' row and ``the right''
row). The cross-over gadget is depicted in Figure \ref{crosscapfig}.

To design the appropriate length of these ladders we need a hardly mentioned
trick (used in all constructions of this family): We pick one particular
drawing of the incidence graph and adjust ladder-lengths for this drawing.
This ensures that at least for one drawing we avoid the problem that some
ladder appears to be too short. To give insight how long paths do we need,
we may draw the graph onto a square-grid with vertices and edge-crossings
in the centers of grid-squares. When replacing vertices by gadgets,
we ensure that each vertex or cross-over is represented by a gadget
of size 1/10 (of the grid-square) in each direction. Then, we have 9/20
in each direction to the boundary of grid-square. We design lengths
of ladders so that even all ladders fit in this gap. I.e. (as we behave
with planarized drawing, we have linearly many edges w.r.t. vertices
and numbers of crossings), i.e., if for each grid-square the ladder
behave with we make this ladder of $500*20/9*n^2$ pairs (which causes
still polynomial blow-up), obviously individual polygons can be so small
that these ladders fit even if all of them should pass through the same
grid-squares.
\begin{figure}[h]
\begin{center}
\includegraphics{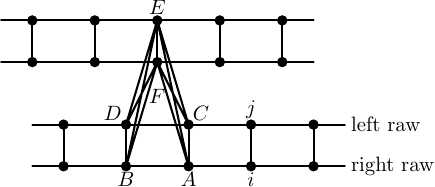}
\end{center}
\caption{This figure shows one cross-over on two ladders. One row
in each ladder is called the left row, the other the right row.
Vertex labels correspond to the labels on the next figure.}
\label{crosscapfig}
\end{figure}
The two pictures in Figure \ref{crosstouch} illustrate how
two ladders touch (resp. cross). In the following paragraphs we refer to 
denotations from that figure.

\begin{figure}
\begin{center}
\includegraphics[scale=0.9]{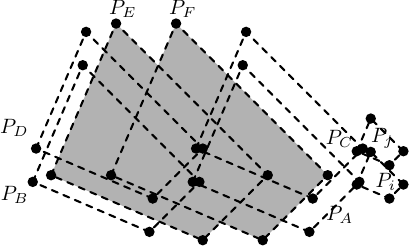} 
\hfill \includegraphics[scale=0.9]{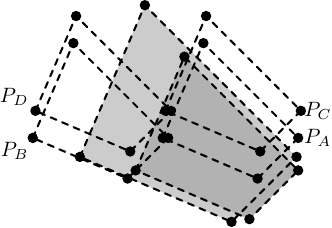}
\end{center}
\caption{The  picture on the left shows how to cross representations
of two ladders in cross-over. The one on the right illustrates the case where they only touch each other.}
\label{crosstouch}
\end{figure}

     More precisely, if we want to
cross two ladders, we represent $P_A$ and $P_B$ by polygons
of the same size crossing only slightly. We create $P_C$ and $P_D$
to be of the same size. So we obtain a quadruple of polygons of
width $(2-\delta)\cdot \hbox{width}(P_A)$ and height
$(1+\varepsilon)\cdot \hbox{height}(P_A)$. We choose a factor
$\mu$ and create polygons $P_E$ and $P_F$ scaled to $P_A$
by factor $(1+\mu)$ and make them cross slightly more
than $P_A$ and $P_B$ do. In the formula, $\delta, \varepsilon$
and $\mu$ are small positive constants depending on prescribed shapes
(and they remain so even in other cases).
Now we use the following fact about
bounding boxes: 
Except for small intervals around the
four points where the polygon touches the boundary of
its bounding box, there is a small stripe inside the
boundary of the bounding box which is disjoint from the polygon.

If we want the ladders to touch only, $P_A,\ldots,P_D$ get
represented in the same manner. We represent $P_E$
by a polygon obtained as follows: 
We take $P_E$ as a copy
of $P_A$ (placed over $P_A$).
Then we shift it to the left (to avoid intersection with $P_i$)
and to the bottom. Then we scale it slightly to intersect
$P_D$ (it can be done by scaling by factor $(1+2\cdot\varepsilon)$).
Now we have a proper representation of the whole gadget except
$P_E$. To represent $P_E$, we take a copy of $P_D$ (placed over
$P_D$). Then we shift it slightly to the right and start
scaling it up to force the right place of intersections with
$P_D$. We certainly may stop scaling before the critical factor
$(2-\delta)$ is reached, as for the factor $(2-\delta)$ we
could place the polygon to intersect $P_A,\ldots,P_E$.
Moreover,
the rightmost corner could be placed below the bottom
intersection of $P_A$ with $P_i$.

Whenever two ladders leading from variables to clauses cross,
we represent this crossing by the cross-over, too.

     After cross-overs the ladder enters {\bf the
clause gadget}. This is represented by a surrounding circle
and a 
structure inside it (see Figure \ref{klauze}).

\begin{figure}
\begin{center}

\includegraphics{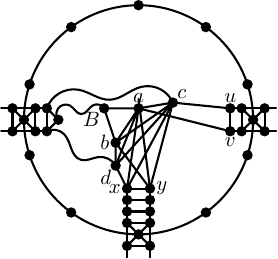}
\end{center}
\caption{The clause gadget, wavy ``edges'' depict arbitrarily long paths.}
\label{klauze}
\end{figure}

     We can easily see that certain problems with the gadget can occur 
only when representing vertices $abcdB$ and their neighbors.
How it is done is illustrated in Figure \ref{repace}. 
We analyze the
existing possibilities depending on the positions of polygons $P_a, P_c$ and $P_x, P_y$.

\paragraph{Case 1.: False/True.} We build this representation
similar to the cross-over described earlier. 
Note that here $P_a$ has the $y$-coordinate of
its left touch-point between $y$-coordinates of left
touch-points of $P_d$ and $P_b$ (while the $x$-coordinates are the same).
For the top and bottom touch-point the situation is similar. 
The polygon $P_c$ is
created as a copy of $P_a$ slightly scaled down and shifted
to the top and to the left, so that it is  covered  from the bottom by
$P_a$ and has  the left and top touch-point with the same
properties as $P_c$.

\paragraph{Case 2.: True/True.} This case is just a special case of the cross-over.

\paragraph{Case 3. True/False.} This case works like in Case 1, but instead of $P_a$
we start with $P_c$.
Then we create $P_a$ as a copy of $P_c$, scale $P_a$ down slightly,
move it slightly upwards, and then more slightly to the left to obtain
the left touch-point of $P_a$'s.
The $x$-coordinate is still larger than
the one of $P_c$, but the top touch-point has the same
$y$-coordinate and its $x$-coordinate is slightly lower.

\paragraph{Case 4.: False/False.} We proceed like in Case 1. After we add
$P_c$, we obtain $P_a$ as a copy of $P_a$ slightly shifted to
the right.
Then we scale it up to obtain an $x$-coordinate of the
left touch-point of $P_a$ that is less than or equal to that of $P_c$, and
a $y$-coordinate of the left touch-point of $P_a$ that is between those
of $P_b$ and $P_d$. 
Then we scale up $P_b$ so that it covers
the intersection point of $P_a$ and $P_c$, and we are
done.

\begin{figure} 
\includegraphics[scale=0.7]{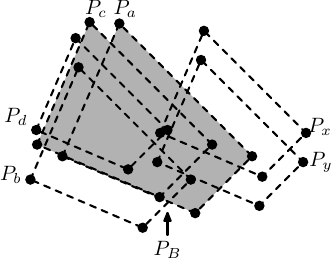}
\includegraphics[scale=0.7]{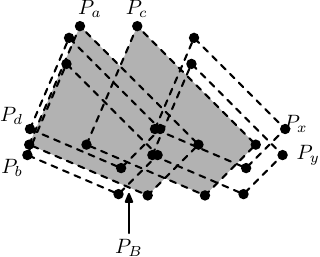}
\includegraphics[scale=0.7]{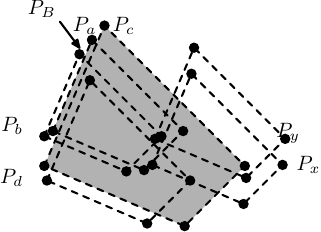}
\includegraphics[scale=0.7]{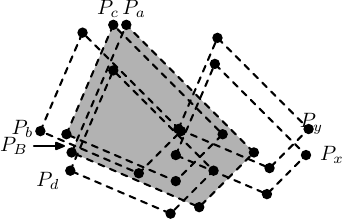}
\caption{Clause gadget and its representation with respect to truth-assignment
for individual variables.}
\label{repace}
\end{figure}

{
The NP-membership of the recognition of $P_{hom}$ graphs has been proved
in \cite{vLvL}. In what follows, we briefly review an argument based on an approach
of \cite{KratMat}, i.e., third 3rd-party tool which, as mentioned earlier, implies the NP-membership stated
in Theorem~\ref{thm:kdirk}.

We have to establish a polynomially-large certificate showing that the
desired representation exists.
For this we use a combinatorial description of the arrangement,
i.e., we guess a description specifying in what order individual sides
of individual polygons intersect. We also need certain information about particular
corners of individual polygons. For individual sides of polygons
we also need to know their directions (for this, 
it is sufficient to keep the index of the direction, i.e., a number
in the set $\{1,..,k\}$). To make the situation formally simpler, instead of segments we consider a description of
the whole underlying lines. Note that a corner of a polygon and the intersection of
boundaries appear here as intersection of two lines.

Now we have to verify the realizability of such an arrangement.
For this, we construct a linear program consisting of inequalities describing the ordering of the intersections
along each side of each polygon. For a line $p$ described by the equation $y=a_px+b_p$, the intersection with line
$q$ precedes the intersection with line $r$ (``from the left to the right'') if
${b_q-b_p \over a_p-a_q}<{b_r-b_p \over a_p-a_r}.$
In the case of prescribed directions ($a_p, a_q, a_r$), we have a linear
program whose variables are the $b$ coefficients. 

This linear program can be solved in polynomial time, which implies the NP-membership for $k\DIRCONVk{L}$. If the directions are not fixed, we use an argument from \cite{KratMat}, that the directions obtained as solutions of the considered linear program are of polynomial size and thus they may also be part of a polynomial certificate. This implies the NP-membership for $k\DIRCONV$ graphs; for $P_{hom}$ graphs we extend the argument as follows:

So far, we have extended the existing linear program by the equations controlling
the ratios of side-lengths for individual polygons. For intersections
of a line $p$
with neighboring sides $q$ and $r$ of a polygon $A$, we add the following
equation:
${b_r-b_p\over a_p-a_r} - {b_q-b_p\over a_p-a_q}=k_p\cdot s_A$.
Here the variable $s_A$ represents the size of a polygon $A$, while $k_p$ refers
to the horizontal length of side lying on a line in the direction of $p$ in a ``unit homothetic polygon"
(i.e., in a polygon with ``unit horizontal length," whose lowest
$x$-coordinate differs from the hightest by one). See Figure \ref{polycert} for
illustration. Note that again
the denominators are constants, so once again we obtain a linear program.
If the shape of polygon $P$ is fixed, we are done. 
Otherwise, if the polygon $P$ is not given, we can regard the directions of its sides as variables and use the same trick as in the proof of Theorem~\ref{thm:kdirk}.}
\end{proof}

Thus we obtain the following strengthening of the results from \cite{vLvLM} showing the existence of
polynomial certificate for the recognition problem when the underlying
polygons are rational.  

\begin{figure}
\begin{center}
\includegraphics{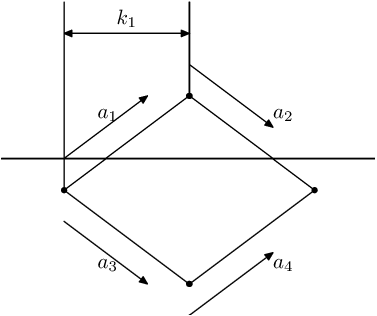}
\end{center}
\caption{Illustration to the verification of the polynomial certificate
of a $P_{hom}$ graph. The parameter $k_1$ is assigned to the whole polygon
and it is designed to verify that all the sides of the polygon were obtained
by scaling the original polygon by the same factor.}
\label{polycert}
\end{figure}

\begin{theorem}\label{thm:NPPhom}
For every fixed $k$, the problem of deciding whether there exists a convex $k$-gon $P$ such that an input graph is in $P_{hom}$, is NP-complete.
\end{theorem}
\noindent This
partially solves Problem 6.3 posed in \cite{vLvLM}, where one asks whether the
recognition of intersection graphs of homothetic convex polygons is
in NP for all convex polygons.

\section*{Acknowledgements} We thank Michael Kaufmann for valuable discussions on the topic.

Parts of the results have been presented at the following conferences: ISAAC 2012 \cite{isaac} and IWCIA 2014 \cite{Brimkov}. 
The work presented at ISAAC 2012 was partially supported by ESF EUROGIGA project GraDR as Czech research grant GA\v CR GIG/11/E023. 
The work presented at IWCIA 2014 was supported by NSF grants No 0802964 and No 0802994.}

\end{document}